\documentclass{birkjour}

\usepackage[english]{babel}
\usepackage{bbm} 
\usepackage[latin1]{inputenc}
\usepackage{amsmath,amsthm,amsfonts,amssymb}
\numberwithin{equation}{section}
\usepackage[final]{epsfig}
\usepackage{color}

\usepackage{mathrsfs}

\usepackage{enumerate}

\newtheorem{theorem}{Theorem}[section]

\newtheorem{lemma}[theorem]{Lemma}

\newtheorem{proposition}[theorem]{Proposition}

\newtheorem*{lemma*}{Lemma}      

\theoremstyle{definition}
\newtheorem{definition}[theorem]{Definition}

\numberwithin{equation}{section}

\usepackage{graphicx,pgf}

\makeatletter
\let\Im\undefined
\let\Re\undefined
\DeclareMathOperator{\Im}{Im \,}
\DeclareMathOperator{\Re}{Re \,}

\DeclareMathOperator{\dist}{dist}

\DeclareMathOperator{\sgn}{sgn}
\newcommand{\Z}{ {\mathbb Z} }

\newcommand{\C}{\mathbb{C}}
 
\newcommand{\G}{\mathcal{G}}
\newcommand{\R}{\mathbb{R}}

\def\ket#1{|#1 \rangle}   
\def\bra#1{\langle #1|}   %

\usepackage{bbm}
\def\id {\mathbbm{1} }

\DeclareMathOperator{\indfct}{\mathbbm{1}}

 \def\Xint#1{\mathchoice
 {\XXint\displaystyle\textstyle{#1}}%
 {\XXint\textstyle\scriptstyle{#1}}%
 {\XXint\scriptstyle\scriptscriptstyle{#1}}%
 {\XXint\scriptscriptstyle\scriptscriptstyle{#1}}%
 \!\int}
 \def\XXint#1#2#3{{\setbox0=\hbox{$#1{#2#3}{\int}$}
 \vcenter{\hbox{$#2#3$}}\kern-.5\wd0}}
 
 \def\dashint{\Xint-}

\def\be{\begin{equation}}
\def\ee{\end{equation}} 

\def\bra#1{\langle #1 |} 
\def\ket#1{| #1 \rangle}
\def\={\  = \  }   
\def\+{\  + \  }   
\def\-{\  - \  }    

\makeatother

\begin{document}

%
%
%
%
%
\title[Partial Delocalization on the Complete Graph]{Resonances and Partial Delocalization on the Complete Graph} 

\author[Aizenman]{Michael Aizenman}
\address{Departments of Physics and Mathematics, Princeton University,\\ Princeton NJ 08544,  USA.}
\email{aizenman@princeton.edu}
  
\author[Shamis]{Mira Shamis}
\address{Mathematics Department, Princeton University, \\ Princeton NJ 08544,  USA.}
 \email{mshamis@princeton.edu}

\author[Warzel]{Simone Warzel}
\address{Zentrum Mathematik, TU M\"unchen, \\
 Boltzmannstr. 3, 85747 Garching, Germany.  }
 \email{warzel@ma.tum.de}


\date{15 May, 2014}

\begin{abstract} 
Random operators may acquire extended states formed from a multitude of   mutually resonating local quasi-modes.  
This mechanics is explored here in the context of the random Schr\"odinger operator on  the complete graph.  The operators exhibits  local quasi modes mixed through  a single channel.   
While most of its spectrum consists of localized eigenfunctions, under appropriate conditions it includes also  bands of states which are delocalized  in the $\ell^1$-though not in $\ell^2$-sense, where the eigenvalues have the statistics of  \v{S}eba spectra.  
The analysis proceeds  through some general observations on the scaling limits of random functions in the Herglotz-Pick class. 
The results are in agreement with  a heuristic condition  for the emergence of resonant delocalization, which is stated in terms of the tunneling amplitude among quasi-modes.

\end{abstract}

\maketitle

\section{Random Schr\"odinger operator on the  complete graph } 

\subsection{The operator and its phase diagram}  

The random Schr\"odinger operator on the complete   graph is given by  the $M \times M$ matrix: 
\begin{equation}\label{eq:H}
H_M  := -  \ket{\varphi_0} \bra{\varphi_0}  + \kappa_M V \, , 
\end{equation}
with 
\be
 \bra{\varphi_0} = (1,1,\dots,1)/\sqrt{M}\, \qquad \mbox{and}  \qquad  \qquad \kappa_M := \frac{\lambda}{\sqrt{2\ln M}}  \, .
\ee 
The rank-one operator  $- \ket{\varphi_0} \bra{\varphi_0}$, which plays  the role of  the kinetic term, is a multiple  of the adjacency matrix on the complete graph of $M$ vertices ($\times (- 1)/M $).
The second term is a multiple of a random potential $V$, which  is a diagonal matrix with independent  
entries $ (V(1),..., V(M)) $ with the common probability density  
\be
\varrho(v) := \exp\left(-v^2/2\right) / \sqrt{2\pi} \, . 
\ee
The gaussian distribution is chosen here mainly for concreteness sake; most of the analysis can be adapted to other distributions with continuous densities, with suitable adjustments in the scale of the coupling constant $ \lambda > 0 $. \\

At $\lambda=0$ the spectrum  of $H_M $  consists of two levels.   The ground state is non-degenerate, at energy $E=-1$ and given by 
the ``extended state'' $\varphi_0$, and the other energy level is  $(M-1)$-fold degenerate, at energy $E=0$.  The degeneracy is split as soon as $\lambda \neq 0$.    For $M\gg 1$, as $\lambda$ is increased  the hitherto degenerate levels spread at rates proportional to $\lambda$, being asymptotically dense in the interval $[-\lambda,\lambda]$.\\ 

This model is studied here as a case study of resonant delocalization. 
The $\delta$ function states which are localized at sites of unusually high values of  the potential $V$ (whose maximum is typically close to $\sqrt{2 \ln M}$) form approximate eigenfunctions, or  ``quasi -modes''.   The kinetic term allows tunneling between such states, and under the right conditions the operator's eigenfunctions  take the form of hybredized mixtures of localized states.  Of particular interest  is the consequent emergence of a narrow spectral band at which the eigenstates are semi-delocalized.     Following is an outline of the results established in this work. \\ 

The energy of the extended state which starts as  $\varphi_0$ changes with $\lambda$ at a much slower rate  ($o(1)$).  At $\lambda=1 +o(1)$  a  first order transition occurs at the spectral edge;  the extended state is passed through an `avoided crossing' by a localized state which is  supported mainly on the minimum of the potential (Theorem~\ref{thm:spectrumandgs}).   As $\lambda$ is  further increased the operator continues to have an extended state at an energy close to $(-1)$, which  is repeatedly passed by localized states as $\lambda$ is increased over the interval $(1,\sqrt 2)$.   The extended state and the localized state passing it hybridize for only a ``brief instant'' on the $\lambda$ scale.  The signature of  that is that for $\lambda < \sqrt 2$,  in energies away from $ E =0$ where the quasi-modes are  initially bunched up, at any a-priori chosen value  of $\lambda$ the operator has only strongly localized states  except for one which is a slightly perturbed version of $ \ket{\varphi_0}$ (Theorems~\ref{thm:localization} and~\ref{thm:localizationplusdeloc}).   \\ 

Another transition happens when $\lambda$ passes the value $\sqrt 2$.  Beyond that more massive hybridization occurs and  a small band of semi-extended states emerges at energies $E = -1 +o(1)$  (Theorem~\ref{thm:resdeloc_rev}).   Similar semi-delocalization is found  at energies close to $E = 0$ for all $\lambda >0$.   \\  

\subsection{The spectral scaling limit}

In discussing the random operator's spectrum in the limit $M\to \infty$ we consider its  blown-up picture under a magnification in which it appears as a random point process with mean spacing of order $1$.   In the regime of resonant delocalization the limit is given by  what we call  the \v{S}eba process,  after the prior appearance  of similar spectra in the \v{S}eba billiard.  It appears also in other contexts, which are mentioned in Section~\ref{sec:seba} where the  process is defined.  \\ 
 
In Section~\ref{sec:general} we present some tools which yield a classification of the different possibilities for the  scaling limits of spectra of similar characteristic equations, and for establishing  convergence.  These may be of independent interest,  involving some general results about the possible limits for the set of roots of a characteristic equation which is expressed in terms of random functions in the Herglotz-Pick class.  \\ 
 
 In Section~\ref{sec:disc} we compare the  results on the emergence of narrow bands or semi-delocalization with a heuristic criterion which is based on a condition for hybridization among two resonating quasi modes.   For this purpose, the standard notion of quasi-modes is enhanced here with a  definition of tunneling amplitude  which is natural for systems with random potential.
   The heuristic criterion is found  to be of relevance also in the present context, even though the hybridization studied here is of somewhat more extensive nature.    Also discussed there are the different notions of delocalization which are of relevance for operators with long range hopping.   \\

\subsection{Relation to past works} 

Our motivation to study resonant delocalization as a mechanism for the formation of bands of extended states was  in part motivated by recent results on  random Schr\"odinger operators on tree graphs~\cite{AW11,AW13}.  The mechanism plays there a role  even in regimes of very low density of states, and it is of interest to understand its role in other systems with rapid growth of volume reached by $n$ steps.      (Tree graphs are also of interest from other perspectives; e.g. the location of the mobility edge is affected by multifractality effects on which more can be found in~\cite{AW11,MG,BRT,LAKS}.) \\

The random Schr\"odinger operator on the complete graph is particularly amenable to analysis, and, by a number of different methods, it has already attracted attention in various contexts:  Anderson localization~\cite{molchanov,SUSY},  quantum chaos~\cite{Seba90,AlSeba,BGS1,BGS2,BBK}, and adiabatic quantum computation \cite{FGGN08}.  
Our discussion overlaps in part with~\cite{molchanov},  which focused more on the localization phase, and~\cite{SUSY} which highlighted a SUSY calculation by which partial localization results were obtained. 
However, these works have not addressed the delocalization phenomena on which we focus.     
In particular, the transition in which many localized modes resonating through a single channel turn into a band of spatially delocalized states, seems to not  have been discussed.   The description of this phenomenon, and the general tools which are presented here, form the main points of this work.

\section{The macro and micro perspective}   
\subsection{Spectral range - on the macroscopic scale}     
    
While our main focus here is on the nature of eigenfunctions,
it is natural to first determine the range of the spectrum $\sigma(H_M) $ and the spectral density.    This can be obtained simply through the following two observations: 
\begin{enumerate} 
\item Since $H_M = - \ket{\varphi_0}\bra{\varphi_0}  + \kappa_M V $ differs from  $\kappa_M V$ by just a rank-one perturbation, the eigenvalues of the two interlace.  Hence the spectral  density of states of $H_M $ is that of its potential part, and given by $\rho(E/\kappa_M)$ (as was already discovered through a somewhat more involved SUSY calculation~\cite{SUSY}).\\ 

\item  Due to well understood large deviations, the values which the random potential $ V $ assumes over the $M$ point set  $\mathcal{K}_M$  typically  span the interval  $ [- \sqrt{2 \ln M},  \sqrt{2 \ln M}]$ up to fluctuations of  order  $1/\sqrt{\ln M}$.   In the normalization selected in \eqref{eq:H}, for fixed  $\lambda$ the two terms of $H_M$ are (typically)  of comparable norms:     
\be \label{eq:spectrum}     
\| T\| \ = \  1 \, \, , \qquad  \| \kappa_M V \| \ = \  \ \lambda + \mathcal{O}\big(\frac{1}{\sqrt{\ln M }}\big)    
\ee     
where the second equality holds in a probability sense.   (This normalization may remind one of a familiar feature  of the random energy model, cf.~\eqref{eq:extremevaluestat} for a precise statement.) 
\end{enumerate}      
          
The emerging picture is summarized is the following statement, in which we employ the notion of  
the Hausdorff distance between two subsets of a metric space, here $ I, J \subset \mathbb{R} $:
\be
 d_H(I,J) := \max\left\{ \sup_{x\in I} \inf_{y\in J} |x-y| \, ,   \, \sup_{y\in J} \inf_{x\in I} |x-y|\right\} \, .
 \ee   
\begin{theorem} [Spectrum and ground state]
\label{thm:spectrumandgs}
For  the sequence of operators $H_M$, with $M \to \infty$ at fixed $\lambda > 0  $: 
\begin{enumerate}
\item  For large $M$ the spectrum $ \sigma(H_M) $ of $H_M$ is typically close, in the  Hausdorff distance $ d_H $, to the  non-random set
\be 
S(\lambda) \  = \  \{ -1 \} \, \cup \, [-\lambda, \lambda]
\ee 
 in the sense that for any $ \varepsilon > 0 $:
\be\label{eq:Hausdconv}
\lim_{M\to \infty} \mathbb{P}\left( d_H( \sigma(H_M), S(\lambda) ) > \varepsilon \right) \ = \ 0 \, .
\ee
\item The ground-state energy and the corresponding ground-state function $ \psi_0$ satisfy with asymptotically full probability:
\begin{align}\label{eq:gs}
\mbox{for $\lambda  \in (0,1) $:} \qquad\quad&
	\min \sigma(H_M) = -1 -\kappa_M^2+ O(\kappa_M^4) \, , \notag \\ 
	& \frac{\| \psi_0 \|_2}{\| \psi_0 \|_\infty} \ = \ \Theta(\sqrt{M} ) \, , \notag \\
\mbox{for $\lambda  \in (1,\infty) $:} \qquad\quad &
	\min \sigma(H_M) = \min \sigma(\kappa_M V) + \frac{\mathcal{O}(1)}{ M }\, , \notag \\ 
	&  1 \leq \frac{\| \psi_0 \|_2}{\| \psi_0 \|_\infty} \ \leq \ 1 + \frac{\mathcal{O}(1)}{\kappa_M \sqrt{M}}  \, .
\end{align}
\end{enumerate}
\end{theorem}

In \eqref{eq:gs}  we employ the following adaptation of the Bachmann-Landau  notion: $ f_M = \mathcal{O}(g_M) $ means that for all sufficiently large $ M $, except for events of asymptotically vanishing probability $ |f_M| \leq C g_M $, and  $ f_M = \Theta(g_M) $ means  that with similar exception $ c g_M\leq f_M \leq C g_M $ for some $ C, c \in (0, \infty) $, with  constants which are independent of the realization of the randomness.  
The proof of Theorem~\ref{thm:spectrumandgs}  is presented here in Appendix~\ref{app:ground_state}.

\subsection{The characteristic equation}  

Further insight can be obtained from the characteristic equations which determines the spectrum.   A  rank-one  perturbation  argument yields (as in~\cite{molchanov}):   
 \begin{proposition}  \label{lem1} For any potential  $V$ which is non degenerate (i.e. $V(x) \neq V(y)$ except for $x=y$)  the spectrum of $H_M$ consists of the collection of energies $E$ for which 
\be \label{eq:ev}
F_M(E) \ := \ \frac{1}{M} \sum_{x=1}^M\frac{1}{\kappa_M V(x) - E} \ = \ 1 \, , 
\ee 
and the corresponding eigenfunctions are given by: 
\be \label{eq:ef}
\psi_{E}(x)   =  \frac{Const.}{ \kappa_M V(x)- E}  \, .
\ee 
\end{proposition}
\noindent (To which it may be added that in the case discussed here degenerate potentials occur only with probability $0$.) 
\begin{proof}
The resolvent expansion $ \frac{1}{H_M - z} \=  \sum_{n=0}^\infty \frac{1}{\kappa_M V - z} \left[ \ket{\varphi_0}\bra{\varphi_0}   \ \frac{1}{ \kappa_M V - z }      \right]^n$
allows to deduce, by standard arguments, that for any $ z \in \C \backslash \R $: 
 \be  \label{eq:res}
 \frac{1}{H_M - z} \= 
\frac{1}{ \kappa_M V - z } \+ [1-F_M(z)]^{-1} \ \  \frac{1}{ \kappa_M V - z } \ \ket{\varphi_0}\bra{\varphi_0}   \ \frac{1}{ \kappa_M V - z }  
 \ee
and,  in particular,
$
\langle \varphi_0 \, , \, (H_M - z)^{-1}  \varphi_0 \rangle \ = \  
(F_M(z)^{-1}  - 1 )^{-1}  $. The spectrum and eigenfunctions of $H_M$ are then read from the poles and residues of its resolvent.
\end{proof}

 Localization properties of the corresponding eigenfunctions \eqref{eq:ef} are expressed here through 
 the ratios $ \| \psi_E \|_p / \|\psi_{E}\|_\infty $, where 
\be 
 \| \psi_E \|_p^p = \sum_{x=1}^M |\psi_{E}(x) |^{p}   \, , 
   \qquad  \|\psi_{E}\|_\infty \ = 
 \  \max_{x} |\psi_{E}(x) |  \, .  
\ee 
and of particular relevance here will be the cases  $p=2$ and $p=1$.   The relation with the usual participation ratio is discussed in Section~\ref{sec:localization}.   \\ 

Our main objective is to identify conditions under which the  eigenfunctions delocalize, and present a mechanism by which  multiple eigenfunction hybridization occurs due to resonances among many local quasimodes of small energy gaps.  \\


\subsection{The microscopic scale}  

As stated in~\eqref{eq:ev},  the eigenvalues of $ H_M $ form a level set of the random function $ F_M(E) $.  
The mean gap between the values of $\kappa_M V$ in the vicinity of  an energy $ \mathcal{E} \in \R $, and thus also between the intertwining eigenvalues of $H_M$, is given by
\be\label{eq:delta}
\Delta_M( \mathcal E) := \frac{\kappa_M}{ M \varrho( \mathcal E /\kappa_M) } 
 \  =  \  \sqrt{2\pi} \kappa_M \, /  M^{1- (\mathcal E/\lambda)^2}~. 
 \ee   
Upon suitable  amplification of the energy scale  
  the spectrum acquires  the form of a  random point process whose mean spacing is of order one.   
Their locations are given in terms of the rescaled energy parameter 
\be  \label{eq:Escale}
u:= \frac{E-\mathcal E_M}{\Delta_M(\mathcal E_M) } \, ,
\ee 
where we also allow for the center of the scaling window ($\mathcal E_M$) to slightly vary with $M$.    
The rescaled eigenvalues $u_{M,n}$ are simple and intertwine with the 
 collection of similarly rescaled values of the random potential $\kappa_M \, V$ in the vicinity of~$  \mathcal E_M$: 
\be\label{eq:potpoint}
\omega_{M,n}\ := \  \frac{\kappa_M V(x_n) -  \mathcal E_M}{\Delta_M(\mathcal E_M)}  \, .
\ee 
These points are  labeled here in increasing order relative to the reference energy $\mathcal E_M$, so that:  
\be \label{Labels}  
... \le \omega_{M,-1} \le \omega_{M,0} \le  0 < \omega_{M,1} < ...\,,  \quad \mbox{and\,  $ \omega_{M,n-1} <  u_{M,n} \le  \omega_{M,n} $.}   \\  
\ee    
  
In discussing the corresponding eigenfunctions, we chose the constant in \eqref{eq:ef} as $Const. =  \Delta_M(\mathcal E_M) =: \Delta_M $, so that:
  \be   \label{norm}
 \psi_n (x_m) \ = \ \frac{1}{\omega_{M,m} - u_{M,n}} \, , 
\ee 
with $(x_n)$ ordered by  the values of $V(x)$, as above.   \\

For a microscopic perspective on the characteristic equation \eqref{eq:ev}, we rewrite $F_M$ in terms of the  rescaled energy parameter~\eqref{eq:Escale}  as  
$F_M(E) \ = \  \frac{1}{M \Delta_M }  \  \sum_n(\omega_{M,n} - u)^{-1}  $.
Splitting the sum into two parts, \eqref{eq:ev}  can be presented as: 
\be \label{eq:sec}  
\boxed{\quad  S_{M,\omega} (u,L)  \ = \  M \Delta_M \  - \    T_{M,\omega} (u,L) \quad }
\ee 
with 
\be   \label{def:SandT}
S_{M,\omega} (u,L)  := \sum_{n}  \frac{  \indfct \left[ |\omega_{M,n}| \le  L  \right]}{ \omega_{M,n}-u} \,, 
\quad 
T_{M,\omega} (u,L)  := \sum_{n}  \frac{  \indfct \left[ |\omega_{M,n}| > L  \right]}{ \omega_{M,n}-u} \, .
\ee 
The cutoff parameter $L= L _M$  will be taken to increase with  $M$ at a rate such that  
\be \label{eq:Lrate} 
1 \ll L_M \ll M^{1- (\mathcal{E}_M /\lambda)^2} / \sqrt{\ln M} \, . 
\ee

The lower bound on $L$ in \eqref{eq:delta} (i.e. the requirement that $L \to \infty$) ensures that  the restricted sum in $S_{M,\omega} (u,L)$  extends over all the terms in  the ``scaling window'', which is described by the limiting point process.  The upper bound aims at keeping the sum in $S_{M,\omega}$ symmetrically balanced with respect to $u=0$, in distributional sense.  The term $T_{M,\omega}(u,L)$ captures the contribution of the singularities  which fall beyond the range described by the scaling limit.  \\  

Within the above scaling window the functions $ S_{M,\omega} (u,L) $ and $T_{M,\omega} (u,L)  $ exhibit quite different dependence on the energy parameter $u$.     In the next section we shall describe some relevant results on the limiting behavior of each of these terms.    This would yield a  short list of possible characteristics of the limiting behavior of the eigenvalue within the scaling window,  and of the corresponding eigenstates.

\section{The  \v{S}eba process}  \label{sec:seba} 

A characteristic equation similar to \eqref{eq:sec} is known to occur also in other contexts, including  \v{S}eba graphs~\cite{Seba90}, singular perturbations of certain chaotic billiards \cite{AlSeba,BGS1,BGS2} and in random matrix theory \cite{FR05}.    In a number of examples, the singularities of $S_{M,\omega}$ converge to a  Poisson point process, as in our case, however the term on the right is replaced by a constant $\alpha \in \R$.  We shall refer to the  collection of the solutions of the corresponding equation as the  $\alpha -$\v{S}eba   process, after ref.~\cite{Seba90}.    This point process would form one of the limiting situations encountered in our  context.    Let us turn to its definition.

\subsection{Definition} 

 Under the mapping which is described by \eqref{eq:potpoint}   
 the collection of rescaled values of the random potential $(\omega_{M,n})$ 
converges in distribution   to a Poisson process of  {\it intensity} $1$ (i.e. mean  density $dx$).  
We refer to the latter as the {\it  standard Poisson process}.    
Its configurations are   countably infinite discrete random  subsets $\omega \subset \R$.\footnote{The scaling limit appears differently   at the spectral edges,  
where the rescaled collection of potential values converges to a Poisson process with intensity  $ e^{\mp u} du $ (cf.~\cite{Bovier}. 
We will however not need to discuss this process; the results presented here for the spectral edges 
 can be obtained through less detailed information.}  \\ 
For any given configuration $\omega\subset \R$, we shall refer to the following function as its  Borel-Stieltjes transform 
   \be  \label{eq:Borel_Stil}  
S_\omega(\zeta) := \lim_{n\to \infty} \sum_{v \in \omega\cap [-n,n]}  
\frac{1}{v -\zeta} \, , 
\ee  
  assuming that  the limit exists for all $\zeta\in \C^+$.  

\begin{proposition}[Theorem~4.1 in \cite{AW_Cauchy}] \label{prop:Stieltjes}
For almost every realization $\omega$ of the  standard Poisson process: 
\begin{enumerate}
\item   The limit \eqref{eq:Borel_Stil} exists almost surely,  simultaneously 
 for all   $ \zeta \in \C^+ $, and 
 \be \label{eq:4.7}
 \lim_{\eta \to \infty} S_\omega(i\eta) = i \pi \,. 
 \ee     
 \item Along the real line the random function $S_\omega(x)$ has only simple poles.  Between any consecutive pair of such a gap $\Delta  v$, the  function increases monotonously from $-\infty$ to $+\infty$, with slope $S_\omega'(x) \ge 1/|\Delta v|^2$.   
\item The thus defined Stieltjes-Poisson random function  $S_\omega(\zeta)$  is a shift-covariant functional of $\omega$, in the sense that   
\be 
 S_{\mathcal{T}_b  \omega}( \zeta ) = S_{\omega}(\zeta+b) 
 \ee  for all $ b \in \R $ and $ \zeta \in \C^+ $, with $\mathcal{T}_b  \omega$ the point configuration shifted to the left by $b$.   
 \end{enumerate}
\end{proposition}   

    \begin{figure}
\includegraphics[width=0.65\textwidth]{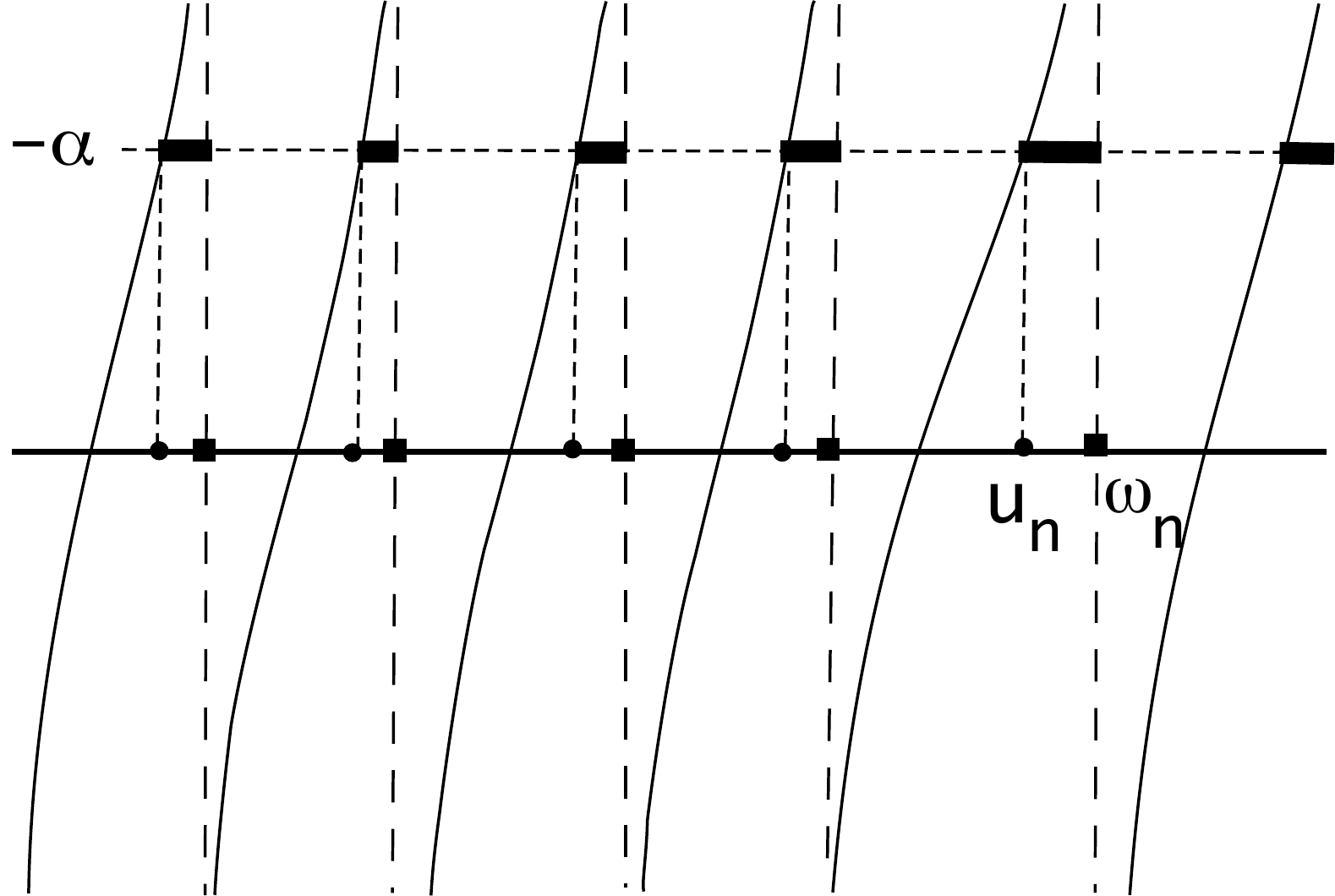}
\begin{center}
\caption{The \v{S}eba process: a schematic depiction of the  solutions of  \eqref{eq:F_alpha} which are discussed in Lemma~\ref{lem:Seba_loc}.  
The function satisfies  $S_\omega(u) \ge - \alpha$ throughout  the intervals  $[u_{n}, \omega_n)$, and $S_\omega(u) \le - \alpha$ throughout   $(\omega_{n-1},u_{n}]$.    This is used in the proof that for $|\alpha| \gg 1$  the  solutions typically lie very close to points of $\omega$, on a side determined by  $sgn(\alpha)$. }
\label{fig:S} 
\end{center}
\end{figure}

\begin{definition} \label{def:Seba}  
Let $S_\omega(\zeta)$  be the Borel-Stieltjes transform of a standard Poisson process $\omega$ whose points are labeled in increasing order relative to $0\in \R$, and for $\alpha \in \R$ let  
$(u_n(\alpha,\omega)=:u_n)$ be the set of solutions of the equation   
 \be  \label{eq:F_alpha}
S_\omega(u)  \ = \ - \alpha \,   
\ee 
ordered as in \eqref{Labels}.

{\it 1.\/}    We refer to the intertwined point process 
$(\{\omega_n, u_n\})$, as the 
{\bf \v{S}eba  process} at parameter $\alpha\in \R$.     

{\it 2.\/}   For any  given $\omega$  and  $\alpha\in \R$, we refer to the points in 
$(u_n)$ 
as the  {\bf \v{S}eba-eigenvalues},  and for each 
$n\in \Z$ 
 regard as the corresponding {\bf \v{S}eba-eigenfunction} 
the function $\Psi_n :  \omega \mapsto \C$ defined by: 
\be  \label{eq:Seba_ef}
 \Psi_n (v) \ :=  \Psi^{(\omega,\alpha)}_n(v) := \  \frac{1}{v-u_n} \, , \qquad   v\in \omega  \, .  \\ 
\ee 
\end{definition}   

The terminology is motivated by the comparison of \eqref{eq:ev} and  \eqref{eq:ef} with \eqref{eq:F_alpha} and 
\eqref{eq:Seba_ef}.    The \v{S}eba-eigenfunctions' norms  
\be 
\|\Psi_{n}\|_p \ := \  \left( \sum_{v\in \omega} |\Psi_{n}(v)|^p \right)^{1/p} \,    
\ee 
satisfy:  
\be 
 \|\Psi_{n}\|_\infty \  = \ \frac{1 }{ \dist \left (u_n(\alpha,\omega), \omega \right)  }
  \, , \qquad   \|\Psi_{n}\|^2_2 \ := \     S'_\omega(u_n) \,.
\ee 

For $p <  \infty$, $\|\Psi_{n}\|_p$  may not yet capture all the relevant information about the eigenfunctions of a finite system whose spectrum within the scaling window the
\v{S}eba process  may approximate, since the finite systems' wavefunctions have  
also weight in regions which asymptotically will be off scale (see Section~\ref{sec:main}).  
Let us nevertheless note  that due to the fact that  $\lim_{n\to \infty} \omega_n /n = 1$ 
 (by the ergodic theorem applied to the Poisson process), one has:    
\begin{lemma}  \label{lem:Seba_ef} 
For any \v{S}eba eigenvalue process at $\alpha \in \R$, with probability one 
all the \v{S}eba-eigenfunctions  are almost surely  $\ell^{1}$-delocalized, 
\be 
\|\Psi_u\|_1 /\|\Psi_{u}\|_\infty \ = \ \infty \, , \qquad 
\ee 
yet also localized in the $\ell^2$-sense, satisfying: 
\be \label{eq:posrec}
  \|\Psi_u\|_2  / \|\Psi_u\|_\infty  \  
 \  \in  (1,\infty)  
\ee 
and $ \|\Psi_u\|_\infty \in  (0,\infty)  $ 
for all $ u \in S^{-1}_\omega(\{ \alpha \})$.\\  
\end{lemma} 
  
   A different situation is found in the limiting case $|\alpha|  \to   \infty$: the \v{S}eba eigenvalues  coalesce then with the  poles of $S_\omega$ and the  \v{S}eba eigenfunctions (once normalized in  $\ell^p$-sense at some $p>1$)  get to be totally localized, each at a single point of $\omega$.  This is quantified in the following  estimate, which in Section~\ref{sec:main} will be used in the discussion of the scaling limits in situations where one finds eigenfunction localization. 
 
\begin{lemma}  \label{lem:Seba_loc}  
The \v{S}eba process at level $\alpha$ satisfies for any $ W , t > 0$: 
 \be  \label{SnUn1}
  \mathbb{P} \left( \max_{n \, : |u_{n}| \le W}   \dist \left (u_n(\alpha,\omega), \omega\cup \{ -W,W\} \right)   \,  \geq   \, t \,  \frac{2\, W}{\max \{ |\alpha|,1\}} \right) \ \leq  
  \ \frac{1}{t}  \, , 
  \ee 
  and for $t>|\alpha|$: 
  \be \label{SnUn2}
   \mathbb{P} \left( \min_{n \, : |u_{n}| \le W}   \dist \left (u_n(\alpha,\omega), \omega \right)   \,  \leq \frac{1}{t} \  \right) 
  \, \leq  \, \frac{2 W }{  t  - |\alpha |}  \, .
  \ee 
 \end{lemma} 
\begin{proof}   
For   any   $q>0$, if    
$\dist (u_n(\alpha, \omega), \omega )    \geq     q  $  for some  $u_n(\alpha, \omega) \in [-W,W]$, then by the monotonicity of the function $S_\omega(u)$ between its singularities  (Fig.~1) it satisfies   $|S_\omega(u)|  \ge  |\alpha|$ throughout  either the interval  $[u_{n}, \omega_n)$ (in case $\alpha <0$) or throughout the interval $(\omega_{n-1} , u_{n}]$ (in case $\alpha > 0$).   This implies that, regardless of the sign of $\alpha$, 
\be     \label{CauchyBnd}
 \int_{-W}^W \indfct [ \sgn (-\alpha)  S_\omega (u)   \geq |\alpha |  ] \, du   \  \geq \    q   ~.   
\ee 
To bound the probability that \eqref{CauchyBnd} occurs, we recall that for each $u\in \R$ the value of $S_\omega(u)$ has the probability distribution of a Cauchy variable with barycentre $i\pi $~\cite{AW_Cauchy}.  Thus: 
\begin{align}  \label{Poisson1}
\mathbb{E} \left( \int_{-W}^W \indfct [ \sgn (-\alpha)  S_\omega (u)   \geq |\alpha |  ] \, du  \right)  & =  \int_{-W}^W  \mathbb{P}  \left (   \sgn (\alpha)  S_\omega (u)  \ \geq |\alpha |  \, \right) \, \, d u  \notag \\ 
&   \leq   \   \frac{2 W}{|\alpha|}  \, . 
\end{align} 
Through the Chebyshev inequality  this allows to conclude that for any $q>0$: 
\be  \label{CauchyProbest}
 \mathbb{P} \left(     
 \int_{-W}^W \indfct [ \sgn (-\alpha)  S_\omega (u)   \geq |\alpha |  ] \, du \  \geq \    q  
         \right) \ \leq  
  \ \frac{2 W}{|\alpha| \,  q} \, . 
\ee 
The choice $q=t \, \frac{2\, W}{\max\{|\alpha|,1\}}  $ yields the first bound claimed in \eqref{SnUn1}.   \\ 
 
In the proof of the second bound it is convenient to employ the function: 
\be 
 S^*_\omega(x) \  := \ \lim_{\varepsilon \to 0} 
 \sum_{v\in \omega} \frac{\indfct [ |v-x|\geq \varepsilon] }{v-x}  
 \ee 
in which  we omit the contribution of the site of $\omega$  which is closest to $x$.  \\ 

Let now $t> |\alpha |$.  We note that if       for  some $u_n(\alpha,\omega) \in [-W,W]$ 
\be  \label{distance}
 \dist \left (u_n(\alpha,\omega), \omega \right)   \le 1/t 
 \ee 
then  
 depending on the sign, $\sigma = \pm$, of the shortest path from  that point to $\omega$, 
 we find that for one of the sites $\omega_n  \in [-W,W] +\sigma \, t^{-1}$   the function 
 $S^*_\omega $  
satisfies, 
\be \label{S*}
\sigma\, \,   S^*_\omega(\omega_n - \sigma \,  t^{-1}) \ \ge \   t - |\alpha|  \, .  \\[2ex]  
\ee 

For a Poisson process the probability that there is a Poisson point within an $\varepsilon$ neighborhood of a given point $x\in \R$,  and the contribution to $S_\omega(x)$ from all other sites of $\omega$, are independent quantities.  
Thus, by a  calculation similar to \eqref{Poisson1}, for either of the two values ($\pm 1$) of $\sigma$:
\be  \label{Poisson2}
\lim_{M\to \infty}   
\mathbb{E} \left( \sum_{v\in \omega \cap ([-W,W]+\sigma \, t^{-1})} \indfct \left[ \sigma S^*_\omega(v-\sigma t^{-1})  
> \ t  -  |\alpha | \, \right]  \  \right)  \ \leq \   \frac{2W}{t- |\alpha |}    . \\[2ex]   
\ee  
Since the probability that \eqref{distance} holds for some $u_n(\alpha,\omega) \in [-W,W]$  is dominated by the mean of the number of such points   we conclude that  for any $t> |\alpha|$: 
\be  \label{head_deloc}
\lim_{M\to \infty}   
\mathbb{P} \left( \min_{n: |u_{n}|<W} \dist(u_n,\omega)  \ < \ t^{-1}  \  \right) 
  \, \leq  \, \frac{2 W }{  t  - |\alpha| } \, .    
\ee  
which completes the proof.
\end{proof}

\section{Scaling limits of Herglotz-Pick functions}    \label{sec:general}

To proceed with \eqref{eq:sec} let us  present three basic results on the possible scaling limits of the random  functions  $S_{M,\omega}(u,L)$ and  $T_{M,\omega}(u,L)$.    Both belong to the Herglotz-Pick class (HP),   
of functions with analytic extension to the upper half plane which assume there only values in $\mathbb{C}^+$  (More on this class of functions can be read in, e.g., \cite{Dono,Duren}.)
By the Herglotz representation theorem any HP function  admits a unique representation in the form: 
\be \label{eq:Herg}
F(z) = a_F z+ b_F + \int \left( \frac{1}{x-z} -  \frac{x}{1+x^2} \right) \  \mu_F(dx)
\ee
with $ a_F \geq 0 $, $ b_F \in \R $ and $ \mu_F $ a Borel measure satisfying: $\int (1+x^2)^{-1} \mu_F(dx)  <  \infty$  
In the cases discussed here the `spectral  measure' $\mu_F$  is of pure-point type, i.e.     $\mu_F$ consists of discrete point masses.   
\\ 

Associated with any open  interval $(a,b)$ is the subclass $P(a,b)$ of  functions which are analytic in $(a,b)$, or equivalently for which $\mu_F( (a,b) ) =0$.

\subsection{The oscillatory part} 

A natural topology for the set of HP functions is that of uniform  convergence on compact subsets of $\C^+$.   The topology is metrizable, with a metric in which this class of  functions forms a complete separable metric space.  Basic properties of the corresponding notion of convergence, and its extensions  to random  functions in this class, are discussed in~\cite{AW_Cauchy}.    Of particular relevance  is the following  implication of \cite[Theorems 3.1 and 6.1]{AW_Cauchy}. 
In applying it to the function $S_{M,\omega}$,  $\widehat \omega_M$ is intended to describe the symmetrically truncated process  
\be  \label{eq:Lest}
\widehat \omega_{M} \ = \  \omega_M  \mathbbm{1}[|\omega_{M,n}| < L_M ]\, , 
\ee 
with $L_M$ satisfying \eqref{eq:Lrate}  (see Lemma~\ref{lem:L}).

\begin{theorem}\label{thm:Sebaconv_rev}
Let  $\widehat \omega_M$ be a sequence of random point processes.
 Suppose that as $M\to \infty$:
\begin{enumerate}[i.]
\item $ \widehat \omega_{M} $ converge in distribution  to the standard Poisson process,  
\item for all $ \varepsilon > 0 $:
\be  \label{eq:as21_rev} 
\lim_{W \to \infty} \limsup_{M \to \infty} \mathbb{P}\left(   \left|   
\sum_{v\in \widehat \omega_M}  \frac{ \mathbbm{1}[|v| > W ] }{v}   \right| \geq \varepsilon \right) \   = 0   \, ,
\ee 
and
\be \label{eq:as22_rev}
\lim_{W \to \infty} \limsup_{M \to \infty} \mathbb{P}\left(    \sum_{v\in \widehat \omega_M}  \frac{ \mathbbm{1}[|v| > W ] }{v^2}   \geq \varepsilon \right) \   = 0   \, .
\ee 
\end{enumerate}
Then 
\be 
\widehat S_{M,\widehat \omega}(z) \ := \  \sum_{v\in \widehat \omega_M}  \frac{ 1 }{v-z}   \,.  
\ee 
converge in distribution to 
the  Stieltjes-Poisson random function $ S_\omega $.
\end{theorem}
The convergence  is in the sense of probability distribution of random functions in the HP class, which are elements of a space whose topology is based on pointwise convergence on $\C^+$ (as  in \cite{AW_Cauchy}). 

\begin{proof}[Proof of Theorem~\ref{thm:Sebaconv_rev}]
Theorem~6.1  in \cite{AW_Cauchy} states that the assertion is implied  by assumption~i. and the requirement that 
for all $ \varepsilon> 0 $:
\begin{align}\label{eq:Thm6.1_rev}
 &  \lim_{\eta \to \infty} \limsup_{M \to \infty} \mathbb{P}(   \left|  \widehat S_M(i\eta) -  i \pi \right| \geq \varepsilon ) \   = 0  \, ,  
 \end{align}
where $i\pi$ is the corresponding distributional limit of $S(i\eta)$.   To  show that this condition is implied by   assumption~{\it ii.}, let us consider separately the contribution to the limit of  the real and imaginary parts of $\widehat S_M(i\eta)-i \pi$.   Condition {\it i.}  and \eqref{eq:as22_rev} entail the distributional convergence of $ \Im \widehat S_M(i\eta) \stackrel{\mathcal D}{\to}  \Im  S(i\eta) $ as $ M \to \infty $ for every $ \eta > 0 $ and hence, using \eqref{eq:4.7}, one may deduce the distributional convergence  of $\Im [ \widehat S_M(i\eta) -  i \pi ]$ to $0$  in the limit seen in~\eqref{eq:Thm6.1_rev}.    To bound the real part, we let $ W = \eta^2 $, and split: 
\begin{align}
\Re \widehat S_M(i\eta)   \ = \ & \sum_n \frac{\widehat \omega_{M,n}}{\widehat \omega_{M,n}^2+\eta^2} \mathbbm{1}[|\widehat \omega_{M,n}| \leq \eta^2 ]  +  \sum_{n}  \frac{ \mathbbm{1}[|\widehat \omega_{M,n}| > \eta^2 ] }{\widehat \omega_{M,n}}   \notag \\ &  - 
 \sum_n \frac{\eta^2}{\widehat \omega_{M,n}(\widehat \omega_{M,n}^2+\eta^2)}  \mathbbm{1}[|\widehat \omega_{M,n}| > \eta^2 ]  \, . 
\end{align}
The last term on the right side is bounded by $ \eta^{-1} \Im \widehat S_M(i\eta) $ and hence, by the just established statement it converges in probability in the double limit $ M \to \infty $ and $ \eta \to \infty $  to zero. The second term converges to $0$ in probability by assumption~{\it ii}.  By assumption~({\it i.}) for each fixed $0<\eta <\infty$ the first term converges in distribution as $ M \to \infty $ to the corresponding sum for the Poisson process $\omega_P$: $\Re  \sum_{v \in \widehat \omega_P \cap [-\eta^2, \eta^2]}  (v-i\eta)^{-1} $ and hence as $ \eta \to \infty $ to zero by \eqref{eq:4.7}.   This concludes the proof of \eqref{eq:Thm6.1_rev}.\\[1ex]
\end{proof} 

\subsection{Linearity at the tail} 

The second preparatory statement addresses possible limits of the tail contributions to Herlotz-Pick functions due to spectral components whose support moves away to infinity.    The following   implies that in their restrictions to 
any fixed window $[-W,W]$ the functions are asymptotically linear, on the relevant scale.   It would be applied to  the tail function   
$[T_{M,\omega}(u,L_M) - M \Delta_M]$ which is of interest here.

\begin{theorem}   \label{thm:linear}
Let $F(z)$ be a  function in the class $P(-L,L)$.  Then for any $W <L/10$ and $u, u_0, u_1 \in [-W,W]$: 
\be   \label{interp} 
\left |  \frac{ F(u)  -   F(u_0)}{ u - u_0 }    -  \frac{ F(u_1)  -   F(u_0)}{ u_1 - u_0 }  \right |  \ 
  \le \   \frac{6W}{L}   \, \, \frac{ F(u_1)  -   F(u_0)}{ u_1 - u_0 }  \,  
  \ee
 \end{theorem} 
\begin{proof}  
The spectral representation \eqref{eq:Herg}  yields
\be \label{eq:H'}
F'(z) = a_F + \int_{|x| > L}  \frac{1}{(x-z)^2}   \  \mu_F(dx)  \  \  ( \geq \  0) ~,  
\ee
from which it follows that for any $u\in [-W,W]$:
\be
\frac{d}{ d u}  \ln F'(u) \ = \   \frac {F''(u)}{F'(u)} \ \le  \    \frac{2}{L-W} 
\ee 
since the bound is valid for the ratio of each pair of corresponding spectral components in the integrals yielding $F''$ and $F'$, the latter being an integral of positive terms.   Hence, for any $x, y \in [-W,W] $:
\be \label{Fprime}
   \frac {F'(x)}{F'(y)}    \  \le \  \exp{\left(\frac{4W}{L-W} \right)}
\ee 
which can be restated in the form:  
\be 
\left |  F'(x) -   F'(y)  \right |  \\ 
  \le \     F'(y)     \,  \left [\exp{\left(\frac{4W}{L-W} \right)}  -1 \right]  \, .
\ee 
Thus, through an application of the mean-value theorem: 
\be   \label{interp_2} 
\left |  \frac{ F(u)  -   F(u_0)}{ u - u_0 }    -  \frac{ F(u_1)  -   F(u_0)}{ u_1 - u_0 }  \right |  \ 
  \le \      \, \, \frac{ F(u_1)  -   F(u_0)}{ u_1 - u_0 }  \,   \left[\exp{\left(\frac{4W}{L-W} \right)}  -1 \right] \, . 
\ee
Using:  $[e^{4/(1-x)} - 1] \le 6x$ for $x\le 1/10$, we get  \eqref{interp}. 
\end{proof}

The asymptotic linearity implies that one has only the following finite list of possibilities for the scaling limits of random functions in the class $P(-L,L)$ (up to a restriction to subsequences, to guarantee consistency).

\begin{definition}  \label{def:3types} 
A sequence of random monotone increasing functions  $K_M(u)$ is said to have a \emph{generalized linear limit}  if the corresponding condition in the following list holds, for all $0<W<\infty$ and $\varepsilon>0$:
\begin{enumerate} 
\item [1.]  {\bf regular linear limit}  \\   for some $a,b\in \R$, $a>0$:  
\be  \label{eq:type1}
\lim_{M\to \infty} \mathbb{P} \left( \max_{u \in [-W,W]}|K_M(u) - (a u + b)| \le \varepsilon  \right) \ = \ 1 \, 
\ee 
\item[$2_{\pm}$.] {\bf singular   $(+\infty)$ or singular $(-\infty)$ limit}  \\ 
  for any $b \in \R$, and $\sigma$ given by the corresponding choice of $\pm1$:  
\be  \label{eq:type2+}
\lim_{M\to \infty} \mathbb{P} \left( \min_{u \in [-W,W]}   \sgn(\sigma) K_M(u)     \, > b   \right) \ = \ 1  \,  
\ee 
\item [3.] {\bf singular limit with transition}  \\  
 there is a ``transition point'', $\tau \in \R$, such that for any $b\in \R$ and $\varepsilon >0$:  
 \be
 \lim_{M\to \infty} \mathbb{P} \left(   \pm K_M(\tau \pm \varepsilon )  \, >  b   \right)  \ = \  1 \qquad .  
\ee
for each value of the $\pm$ sign. 
\end{enumerate} 
\end{definition} 

In essence, the above is the list of possible scaling limits of  linear  functions, $\alpha(u) = a u+b $ with $a\ge 0$, allowing the parameters to assume diverging values.   
The distinction here is between the  cases: 
\begin{enumerate} 
\item both $a$ and $b$ assume a finite limit,  
\item  one of the parameters dominates over the other (either $|a|\ll |b|$, or $|a|\ll |b|$),  with $\pm = \text{sign}(a/b)$), 
\item   the two diverge but the ratio $a/b$ has a finite limit. \\    
\end{enumerate}  

One may note that for functions to which Theorem~\ref{thm:linear} applies, just pointwise convergence, in the sense that 
 \be  \label{eq:Klimit} 
\ {\mathcal D} - \lim_{M\to \infty}  K_{M,\omega}(u)  \  \in \ \R\cup\{-\infty\}\cup \{+\infty\}
\ee 
for a dense collection of $u\in \R$, implies the uniformity which is required  in Definition~\ref{def:3types}, cases (1) and (3). (The argument can be seen in the proof of Lemma~\ref{lem:genlinlimit} below.)

\subsection{Convergence to the \v{S}eba process} 

Theorems~\ref{thm:Sebaconv_rev} and \ref{thm:linear} carry the following  implication  
for the limiting behavior of the solutions of an equation of the form \eqref{eq:sec}.

\begin{theorem} \label{thm:GenCrit}
Let  $\widehat \omega_M $ be  a sequence of point processes satisfying the assumptions {\it i} and {\it ii.} of Theorem~\ref{thm:Sebaconv_rev} and $\widehat S_{M,\widehat \omega}(u)$ the sequence of the related functions 
\be
\widehat S_{M,\widehat \omega}(u) \ = \    \sum_{v\in \widehat \omega_M}  \frac{ 1 }{v-u}   \, , 
\ee 
and let $K_{M}(u)$ be a sequence   of random Herglotz-Pick  functions with generalized linear limiting behavior.
Then the intertwined point process consisting of the support of $\widehat \omega_M$ and  the solution set of equation 
\be \label{eq:TheEq}
\widehat S_{M,\widehat \omega}(u) + K_{M}(u) \ = \ 0 
\ee
has the following limit:
\begin{enumerate} 
\item If $K_{M}$ has a  regular linear limit   given by a constant function 
(i.e. type $1$ limit  in the sense of Theorem~\ref{thm:linear} with $a=0, |b| < \infty$),  
then the intertwined point processes converges  to the \v{S}eba  process at the corresponding level ($  \alpha = b $).  
\item If $K_{M}$ has a singular $(+\infty)$ or singular $(-\infty)$ limit   (type  $2_\pm$), then the solution set coincides asymptotically with $\widehat \omega$.   
\item If the limit is singular with transition (type $3$) then the solution set converges to the limit of $\widehat \omega \cup \{\tau \}$, where the added point is the transition point of $K$.  
\end{enumerate} 
\end{theorem}

It may be added that  the outcome is not affected by possible correlations between the two random functions $H_M$ and $K_M$.
 
\begin{proof}  
1. Let $W > 0$ be a fixed, but arbitrarily large number. 
According to Theorem~\ref{thm:Sebaconv_rev}, for any fixed $u \in \mathbb{R}$, $\widehat S_{M,\widehat \omega}(u)$ converges  in distribution to the 
Stieltjes-Poisson random function $S_\omega(x)$.  
By Skorokhod's representation theorem \cite[Section~1.6]{Bil}, $\widehat S_{M,\widehat \omega}(u)$
and $S_\omega(u)$ can be realized on the same probability space (i.e. the two probability measures ``coupled''), so that $\widehat S_{M,\widehat \omega}(u) \to S_\omega(u)$ with probability one.  The coupling can be implemented simultaneously for all $u \in [-W, W] \cap \mathbb{Q}$ (which forms a countable collection of sites), i.e. 
with probability one:
\be\label{eq:ratconv_rev}
 \forall  u \in [-W, W] \cap \mathbb{Q} : \quad \widehat S_{M,\widehat \omega}(u) \to S_\omega(u)~.
\ee
Let us show that with probability one
\be\label{eq:allconv_rev} \forall u \in [-W, W] \setminus \omega : \quad 
\quad  \widehat S_{M,\widehat \omega}(u) \to S_\omega(u)~.
\ee
By assumption i.\ 
By the assumed convergence of $,\widehat \omega_M$ to the Poisson process, and an additional application of the Skorokhod representation theorem, we may additionally  assume that with probability one $\widehat \omega_{M}  \to \omega $.   Any point $u \in [-W, W] \setminus \omega$ has a neighborhood $I_u$ which is disjoint from $\omega$. Therefore 
any smaller neighborhood $I'_u \subset I_u$ is disjoint from all $\widehat \omega_M$ for all $M\geq M_u$ sufficiently large. The functions $\widehat S_{M,\widehat \omega}$ are monotone 
increasing on $I'_u$ (for $M \geq M_u$), and converge
to a monotone increasing function $S_\omega$ on 
a dense subset of $I'_u$ by~\eqref{eq:ratconv_rev}.  
Since $\widehat S_{M,\widehat \omega}$ is continuous on $I'_u$, this implies that $\widehat S_{M,\widehat \omega}$
converge to $S_\omega$ uniformly on $I'_u$, and in particular $\widehat S_{M,\widehat \omega}(u) \to S_\omega(u)$, proving
\eqref{eq:allconv_rev}.

Let $\{u_{M,n}\}$ be the finite collection of solution of the equation \eqref{eq:TheEq} within   $[-W, W]$.   
According to \eqref{eq:allconv_rev}, for any $\epsilon > 0$ there exists $M_\epsilon>0$ such that for all $ n $ and all $M \geq M_\epsilon$:
\be
| K_{M}(u_{M,n}) - b| < \epsilon~. 
\ee

Between any pair of consecutive poles $v_{M}$, $v_{M}'$ of $\widehat S_{M,\widehat \omega}$ the function 
$\widehat S_{M,\widehat \omega}(u) + K_{M}(u) $ is continuous and monotone increasing from 
$(-\infty)$ to $(\+\infty)$.  Therefore Eq.~\eqref{eq:TheEq} has 
a  unique solution $u_{M,n}$ there. 
By    
item~2 of Proposition~\ref{prop:Stieltjes} (which bounds the derivative of $\widehat S_{M,\widehat \omega}$ from below), for all $ M \geq M_\epsilon $:
\be
|u_{n} - u_{M,n}| \leq \epsilon |v_{M}-v_{M}'|^2 ~. 
\ee
In combination with assumption~i., this implies the stated convergence to the \v{S}eba process.\\ 

2.  For $|b| \gg 1$, the above argument implies that the solutions of  \eqref{eq:TheEq} are at distances of order $O(|b|^{-1})$ from the sites of $\omega$ (the sign of $b$ determining the side from which the approach happens).   Direct adaptation of the argument to  the limit  $|b| \to \infty$ implies the statement in case $2$. \\ 

3.  In case $3$, the functions   $K_{M}(u_{M,n})$ exhibit a sharp transition whose location converges to $\tau $ from   $(-\infty)$ to $(+\infty)$ (as asymptotic values).   Thus, by an extension of the above argument,  the solution set converges to $\omega\cup \tau$.  
\end{proof}

\section{Relevant estimates}   \label{sec:mainresults}

\subsection{The oscillatory term $S_{M,\omega}(u)$} 

The following estimate shows that under the condition \eqref{eq:Lrate}, the assumptions of  Theorem~\ref{thm:Sebaconv_rev} apply to the function $S_{M,\omega}(u,L_M)$ which is defined in \eqref{def:SandT}.   Consequently, in the scaling limit this component of the characteristic equation converges to the Stieltjes-Poisson random function.

\begin{lemma} \label{lem:L} 
Let $\mathcal E_M$  be a sequence of energies in $(-\lambda,\lambda)$,  $L_M$ a sequence of cutoff values satisfying 
\be \label{eq:Lbound}
\lim_{M\to \infty} \, s_M     =  0\, ,  \qquad s_M  := \ L_M \,  \frac{\sqrt{ \ln M } } {M^{1-(\mathcal E_M/\lambda)^2 } }  
\ee 
and $\omega_{M,n}$  the process of rescaled potential values defined relative to  $\mathcal E_M$ 
 through the scaling relation \eqref{eq:potpoint}.  
Then the quantities 
\be \label{eq:Y}  
Y_{M,W} = \sum_{n} \frac{ \indfct [W < |\omega_{M,n}| < L_M] } {\omega_{M,n} } \, 
\ee 
 are of mean and variance satisfying, for each  $W<\infty$: 
 \be 
\lim_{M \to \infty} \mathbb{E} \left[ Y_{M,W} \right] \ = \   0  \, ,   \qquad  
\lim_{M \to \infty}  {\rm Var} (Y_{M,W})  \  \le \    \frac{\text{C}}{W}
\ee  
with a uniform $\text{C} < \infty$. 
\end{lemma} 
\begin{proof} 
Taking the expectation value and expressing it in terms of the macroscopic variables, one gets 
\begin{align}    \label{eq:L1}
\mathbb{E} \left[ Y_{M,W} \right] &=   \frac{M \Delta_M}{\kappa_M}  \, 
\int _{ W\Delta_M/\kappa_M}  ^{L_M \Delta_M/\kappa_M} \, 
    \frac{\varrho(\frac{\mathcal E_M}{\kappa_M}+v) - \varrho(\frac{\mathcal E_M}{\kappa_M}-v)}{v} \, dv \notag \\
    &\le    C \, \frac{M\Delta_M}{\kappa_M}  \,   \max_{|x|\leq L_M \Delta_M} \varrho\left(\frac{\mathcal E_M+x}{\kappa_M}\right) \  
   \, \frac{|\mathcal E_M| + |x|}{\kappa_M}
    \,  \frac{L_M \Delta_M}{\kappa_M}  \notag \\[2ex] 
    &  \le  C\,    \frac{\Delta_M}{\kappa_M^ 2} \, L_M   \left(1+  s_M   \right)^2\ = \ C\ s_M \,  \left(1+  s_M   \right)^2 ,     
\end{align}       
where $ C< \infty  $ is a constant and use was made of \eqref{eq:delta} and the observation that the above shift by $x$ does not cause a significant change in  $\log \varrho$.    

To estimate the variance,  we note that despite  our  labeling  convention the   sum  in \eqref{eq:Y}  
can also be regarded as over  $M$ independently and identically distributed variables $\omega_n$ with values in $ \R$.   The variables' independence easily implies for the sum's variance: 
\begin{align}  \label{eq:L2}
{\rm Var} (Y_{M,W})  & \leq  \  
   \mathbb{E} \left[  \sum_{n} \frac{ \indfct [W < |\omega_{M,n}| < L_M] } {|\omega_{M,n} |^2}  \right] 
\notag \\
& \leq \frac{M \Delta_M^2}{\kappa_M^2} \int _{ W\Delta_M/\kappa_M}  ^{L_M \Delta_M/\kappa_M}    \frac{\varrho(\frac{\mathcal E_M}{\kappa_M}+v) + \varrho(\frac{\mathcal E_M}{\kappa_M}-v)}{v^2} \, dv \notag \\
& \leq  C \ \frac{M \Delta_M^2}{\kappa_M^2} \max_{|x|\leq L_M \Delta_M} \varrho\left(\frac{\mathcal E_M+x}{\kappa_M}\right) \frac{\kappa_M }{W \Delta_M}   \notag \\
& \leq  \ \frac{C}{W} \,  \left(1+  s_M   \right) ,  
\end{align}
where again $ C< \infty  $ is a constant.
\end{proof}

\subsection{The tail term} 

The other component of the characteristic equation \eqref{eq:sec}  
is the function 
\be  \label{def:R} 
R_{M,\omega}(u) \ := \  T_{M,\omega}(u,L_M)   - M \Delta_M ~. 
\ee   
Let us first consider its mean value.

\begin{lemma}\label{lem:meanT}
Let $\mathcal E_M$  be a convergent sequence of energies, with  $\lim_{M\to \infty} \mathcal E_M  \ \in \ (-\lambda,\lambda)$, and  $L_M \to \infty $ a sequence of cutoff values satisfying~\eqref{eq:Lbound}. Then
for any $ u \in \R $:
\begin{multline} 
\lim_{M\to \infty}   \,  \left( \mathbb{E}\left[T_M(u,L_M) \right]  - M \Delta_M \widehat{\varrho}_M( \mathcal E_M + u \Delta_M)   \right)  \\ 
= \ - \lim_{M\to \infty}   \, \mathbb{E}\left[S_M(u,L_M) \right]   \ =   \ 0 \,    
\end{multline} 
with  $\widehat{\varrho}_M$ the function 
\be \label{eq:Hilb} 
\widehat{\varrho}_M( \mathcal E) :=   \dashint\frac{\varrho(v) \, dv}{ \kappa_M  v-  \mathcal E } \, , 
\ee
the dash in $\dashint$ indicating a principal-value integral (i.e., up to scaling, the Hilbert transform of $\varrho$).  
\end{lemma}

\begin{proof} 
The first equality follows directly from the definition of the different terms.  Taking advantage of that, we turn to estimate the  expectation value $\mathbb E (S_M(u,L_M) )$.  It can be written as a sum of   two terms.  One of them is
\begin{align}  
 M \Delta_M \ \varrho\left(\frac{\mathcal{E}_M+ u \Delta_M}{\kappa_M}\right)  \dashint  \frac{ \mathbbm{1}[|\kappa_M v - \mathcal{E}_M| \leq \Delta_M L_M ] }{\kappa_M v - \mathcal{E}_M - u \Delta_M}  dv \notag \\
 = \frac{\varrho\left(\frac{\mathcal{E}_M+ u \Delta_M}{\kappa_M}\right) }{\varrho\left(\frac{\mathcal{E}_M}{\kappa_M}\right) } \ln \frac{L_M-u}{L_M} \to \ 0 \, .  
\end{align}
This tends to zero since $ L_M \to \infty $, and the effect of the shift in the argument of $\varrho$ is controllable as in \eqref{eq:L1}.

The second term in modulus equals 
\begin{align}
&M \Delta_M \, 
\left| \int \indfct[|\kappa_M v - \mathcal{E}_M| \leq \Delta_M L ] \;   \frac{\varrho(v)- \varrho((\mathcal{E}_M+ u \Delta_M)/\kappa_M) }{\kappa_M v - \mathcal{E}_M - u \Delta_M} dv \right| \notag \\
&\le   C \, \frac{M\Delta_M}{\kappa_M}   \max_{|x|\leq L_M \Delta_M} \varrho\left(\frac{\mathcal E_M+x}{\kappa_M}\right)  
   \, \frac{|\mathcal E_M| + |x|}{\kappa_M}
    \frac{L_M \Delta_M}{\kappa_M}  \leq  C    \frac{L_M \sqrt {\ln M} }{M^{1-( \mathcal E_M/\lambda)^2 } } \, .   
\end{align}       
The last inequality is the same as in \eqref{eq:L1}. The rights side vanishes in the limit $ M \to \infty $ by~\eqref{eq:Lbound}.
\end{proof}

We shall  be interested in locating the energies at which $\mathbb{E}\left[ R_{M,\omega}(u) \right] $ is of order $1$ (which is where significant wave function hybridization may occur).   In terms of the mean, this corresponds to:  
\be  \label{eq:Tmean} 
 \ M\, \Delta_M   \big|    \widehat{\varrho}_M( \mathcal E)  - 1  \big|  \ \lesssim \ 1\, .
\ee
Away from $\mathcal E =0$ (i.e. where $M \Delta_M \to \infty$) this condition is satisfied  only for $\mathcal E$ in the vicinity of the solutions of
\be \label{eq:R=0}
\widehat{\varrho}_M( \mathcal E)  \ - \ 1   \ = \ 0  \, . 
\ee  
Looking closer, one finds that: 
\begin{enumerate}[i.]
\item  for each $\mathcal E\neq 0$:  
\be \label{eq:rhoE}
\lim_{M\to \infty} \widehat{\varrho}_M( \mathcal E)   = -1/\mathcal E 
\ee  
\item   the principal-value cutoff rounds off the singularity at $\mathcal E=0$, so that for finite $M$   the functions 
$ \widehat{\varrho}_M( \mathcal E) $ 
is continuous, changing sign  at $\mathcal E=0$,      
\item   for $\mathcal E \notin \{-1, 0\}$,  the factor 
$M \, \Delta_M$ causes $R_M $ to diverge.   
\end{enumerate} 
More precisely: 

\begin{lemma}  For large enough $M$, Eq.~\eqref{eq:R=0} has 
 exactly two  solutions, one 
  in the vicinity of $  \mathcal{E} = -1 $, which will be denoted  
  $   \widehat{\mathcal E}_{M}^{(-1)} $, 
  and the other  in the vicinity of  $ \mathcal{E} = 0 $, denoted  here as $  \widehat{\mathcal E}_{M}^{(0)} $.  
  For $M \to \infty$ these behave as: 
\begin{eqnarray}  \label{eq:asympE}
  \widehat{\mathcal E}_{M}^{(-1)}  & = &    - 1 -  \kappa_M^2 +  \mathcal{O}\left(\kappa_M^4\right)  \notag \\ 
    \widehat{\mathcal E}_{M}^{(0)}  & = &    - \frac{\kappa_M^2}{2\sqrt{\pi}} \  + \   \mathcal{O}(\kappa_M^4) \, ,  
\end{eqnarray}  
\end{lemma} 
\begin{proof} 
The Hilbert transform of the Gaussian probability density $ \varrho(v) $ is an odd continuous function which satisfies the asymptotics:
\be \label{eq:HTrafo}
 \dashint\frac{\varrho(v) \, dv}{   v-  \xi}  =
    \begin{cases}  
       - 2 \sqrt{\pi} \left[ \xi - \frac{2}{ 3} \xi^3  + O( |\xi|^{5} )\right] & \mbox{for small $\xi$,} \\[1em] 
       -  \xi^{-1}-  \xi^{-3} + o(  \xi^{-3} ) &  \mbox{for large $\xi$.}
    \end{cases} 
\ee 
Since $ \kappa_M \to 0 $ the claim immediately follows therefrom.
\end{proof}

To reach beyond the mean, and in particular apply Theorem~\ref{thm:GenCrit} to the  function $R_{M,\omega}(u,L)$
we need information on the function's fluctuations.    The following  estimates will be of help.

\begin{lemma}\label{lem:tail_var}  
For any  $\lambda>0$, $\mathcal E_M \in (-\lambda , \lambda)$,  $1 \  <\ L_M \ <\    M^{1-(\mathcal E_M/\lambda)^2}  {\Large /} \sqrt {\ln M} $,  and $|u| < L_M/2 $:  
\begin{eqnarray}  \label{eq:Tvar2}  
Var(R_M(u))  & = &  \mathbb{E} \left( \big| T_M(u,L_M)  - \mathbb{E}(T_M(u,L_M))   \big |^2  \right)   \notag \\ 
&   \leq  &  
 \mathbb{E} \left(  \frac{d}{d \, u}  T_M(u,L_M)  \right)         
\end{eqnarray} 
and:  
\be \label{eq:Tprime}  
C_1\  \frac{M^{2(\mathcal E_M/\lambda)^2 -1} }{\ln M  } \  \leq \  \mathbb{E} \left(   \frac{d}{d \, u}  T_M(u,L_M) \right)  
\  \leq  \   
 C_2\ \left[ \frac {1}{L}   + M^{2(\mathcal E_M/\lambda)^2 -1}  \right]   
\ee
with   uniform  constants  $0 < C_1 < C_2  < \infty $.     
\end{lemma} 
\begin{proof}  
The first equality in \eqref{eq:Tvar2} holds since $R_{M,\omega}(u)$ and $T_{M,\omega}(u,L_M)$ differ by a constant.   
For the variance bound, we start as in the first step in \eqref{eq:L2}: 
\begin{eqnarray}  \label{eq:101} 
\mathbb{E} \left( \big| T_M(u,L_M)  - \mathbb{E}(T_M(u,L_M))   \big |^2  \right)  &  \leq  &   
 \mathbb{E} \left(  \sum_{n}  \frac{  \indfct \left[ |\omega_{M,n}| > L_M  \right]}{ |\omega_{M,n}-u|^2 }  \right)  \notag \\ 
 & = &  \mathbb{E} \left( \frac{d}{d \, u}  T_M(u,L_M)  \right)  \, .  
\end{eqnarray}  
The sum in  the right is estimated for any $ |u|\leq W < L_M $ using the elementary inequality 
\be \label{eq:vergle}
\frac{1}{1+ W/L_M} \leq  \frac{|\omega_{M,n}|}{|\omega_{M,n} - u |} \leq \frac{1}{1- W/L_M} 
\ee
since $ |\omega_{M,n}|\geq L_M $. It hence remains to estimate  $ \frac{d}{d \, u}  T_{M,\omega}(0,L_M) $ which can be expressed in terms of the 
macroscopic variables (of \eqref{eq:Escale} and \eqref{eq:potpoint}): 
\begin{multline} 
\mathbb{E} \left( \frac{d}{d \, u}  T_M(0,L_M)   \right)  \ = \  \frac{M \, \Delta_M^2}{\kappa_M^2} \, 
\int \frac{\indfct  \left[|v-\mathcal E_M /\kappa_M| > \delta_M \right] }{ |v -  \mathcal E_M/\kappa_M |^2} \  \varrho(v) \, dv
\end{multline} 
 with $\delta_M := L_M \, \Delta_M/ \kappa_M = \sqrt{2\pi} L_M /\,  M^{1-(\mathcal E_M/\lambda)^2}$.
 The above integral is estimated in Appendix~\ref{app:est}.   Using  the upper bound in \eqref{eq:est} 
 and  the relations \eqref{eq:delta} we get: 
\begin{align} 
\mathbb{E} \left( \frac{d}{d \, u}  T_M(0,L_M)   \right)    &  \leq  
  C  \frac{M \, \Delta_M^2}{\kappa_M^2}  \left[ \frac{\varrho(\mathcal E_M/\kappa_M)}{\delta_M} + \frac{1}{(1+| \mathcal E_M /\kappa_M|)^2}\right]  \notag \\
&  \leq   \ C  \left[  \frac{1}{L_M} + \frac{M \, \Delta_M^2}{\kappa_M^2}   \right]   \, .
 \end{align}  
Likewise, the lower bound in \eqref{eq:est}  yields 
\be
\mathbb{E} \left( \frac{d}{d \, u}  T_M(0,L_M)  \right)    \  \geq   \  
  c \  \frac{M \, \Delta_M^2}{(\kappa_M+| \mathcal E_M |)^2} \  \geq    \   
  c  \   \frac{M^{2(\mathcal E_M/\lambda)^2 -1}  }{ \ln M }\, . 
\ee
\end{proof} 

The variance bounds of Lemma~\ref{lem:tail_var}  show  a qualitative transition when $|\lambda/\mathcal E_M|$ crosses the value  $\sqrt 2$.   For  $|\lambda/\mathcal E_M| > \sqrt 2$ , the fluctuations of the tail function $T_{M}(u,L_M)$  diverge with $M$.  Nevertheless, even there, the divergence rate  is slower than  the typical value of the function's derivative.   This would have the significant implication that the energy at which $R_{M}(u,L_M) $ changes sign would  be  deterministic even on the microscopic scale. 

\medskip 

\subsection{Scaling limit of the function $R_{M,\omega}$}   \label{sec:regions}

Considering  the effects of fluctuations we arrive at the following characterization of the scaling limit of $R_{M,\omega}$ in the different regimes, as it is observed within windows $[-W,W]$ of  fixed, though arbitrarily large size.  

\begin{lemma} \label{lem:genlinlimit}
Let $\mathcal E_M$  be a sequence of energies with 
$ \lim_{M\to \infty} \mathcal E_M = \mathcal E  $, $|\mathcal E  | < \lambda$, 
and  $L_M $ a divergent sequence of cutoff values satisfying~\eqref{eq:Lbound}.  Then, in the regimes listed below, 
the random  function 
$R_{M,\omega}(u) $ (defined in \eqref{def:R})  
has the following limiting behavior.  \\ 
\begin{enumerate} [i.]
\item If $\mathcal E \in (\lambda, \lambda) \backslash  \{-1,0\}$, then the limit is $(\pm)\, \infty$  \emph{singular}, the sign being   
$\pm  = \sgn  [\mathcal E (\mathcal E +1)  ]$,  with the random function satisfying, for any $W < \infty$ and $ b > 0 $:  
\be  \label{eq:Ri0}
\lim_{M\to \infty} \mathbb{P} \left( \min_{u \in [-W,W]}    \left   \{ 
-\sgn\big [\mathcal E (\mathcal E +1) \big ] \, R_M(u)   \right \} 
  \, > b   \  \frac{M^{(\mathcal E_M/\lambda)^2} }{\ln M } \right)  \ = \ 1 
\ee   
  
\item  If $\mathcal E =0$, and  
\be  \label{eq:ii}
\sup_M \frac{|\mathcal E_M|}{\kappa_M^2} \  < \  \infty 
\ee  
 then  for any $W < \infty$ and $\varepsilon >0$:
 \be  \label{eq:Ri}
\lim_{M\to \infty} \mathbb{P} \left( \max_{u \in [-W,W]}   |R_M(u)  | 
  \, <  \varepsilon   \   \right)  \ = \ 1 
\ee

\item    In the vicinity of  $\widehat{\mathcal E}_{M}^{(-1)} $   
 the nature of the limit depends on $\lambda$:  \\[0.2 ex] 
   \begin{enumerate} 
    \item if $\lambda < \sqrt 2$:  \quad
any sequence with 
\be   \label{eq:iiia}
\tau \ =\  \lim_{M\to \infty}  \left[\widehat{\mathcal E}_{M}^{(-1)} - \mathcal{E}_M \right] /\Delta_M 
\ee 
the function has a \emph{singular limit with transition at $\tau$},  and satisfies for any $ b > 0 $:
\be  \label{eq:Riia}  
\lim_{M\to \infty} \mathbb{P} \left( \min_{u: \delta_M < |u-\tau| < W}    \left\{ 
  \frac{R_M(u)}{u - \tau } \,   \right \} 
  \, > \, b  \,   \  \frac{M^{2(\mathcal E_M/\lambda)^2 -1} }{\ln M  }   \right)  \ = \ 1  \\[1.5ex]  
\ee 
with $\delta_M  :=   (\ln M)^2 / M^{(\mathcal E_M/\lambda)^2 -1/2} $.\\ 
 
    \item If $\lambda > \sqrt 2$:  \quad for sequences with 
\be   \label{eq:iiib}
\alpha  \ = \   \lim_{M\to \infty}   M \Delta_M   \left[\mathcal{E}_M - \widehat{\mathcal E}_{M}^{(-1)}\right] 
\ee  
$R_{M,\omega}(u)$  converges to the constant function of value $\alpha$ (in a sense similar to \eqref{eq:Ri}).
   \end{enumerate} 

\end{enumerate} 
\end{lemma} 
 
Thus, conditions for resonant delocalization exist at $\mathcal E =0$ at all $\lambda$, and in the vicinity of $\mathcal E =-1$  for $\lambda > \sqrt 2$.    
Another heuristic perspective on the relevant condition is  given below in Section~\ref{sec:disc}.

\begin{proof} \noindent {\it i.}   By \eqref{eq:rhoE}, for any $\mathcal E \in (-\lambda, \lambda) \backslash  \{-1,0\}$, and  $u\in \R$, at $M$ large enough:  
\be
\left[  \widehat \varrho_M(\mathcal{E}_M+ u \Delta_M)-1 \right ]   \ \geq \    \left|1+ 1/\mathcal E \right|   / 2  \ > \  0  \, .
\ee 
Using Lemma~\ref{lem:meanT}, we conclude that for any given $W<\infty$, 
\be 
- \sgn( 1+ \mathcal E^{-1})  \, \mathbb E \left( R_M(\pm W) \right)  \  \ge  \  M \Delta_M   \frac{|  1+ \mathcal E^{-1}| }{2}  \ \ge 
\ C(\mathcal E) \ \kappa_M\ M^{(\mathcal E_M/\lambda)^2}    ~.  
\ee 
By Lemma~\ref{lem:tail_var} the mean square fluctuations are of smaller  order of magnitude, by a factor which is bounded by 
$\ln M / \sqrt M$.   
 This, combined with the monotonicity of the function over $[-W,W]$ implies \eqref{eq:Ri0}.\\

\noindent {\it ii.} Here \eqref{eq:ii} and \eqref{eq:HTrafo} implies that
\be
\lim_{M\to \infty} M \Delta_M \left( \widehat \varrho_M(\mathcal{E}_M)-1 \right) = 0 \, .
\ee
The bound   \eqref{eq:ii} on $\mathcal E_M$ combined with the  bounds of Lemma~\ref{lem:meanT} and Lemma~\ref{lem:tail_var},  guarantee that the fluctuations and the derivative  in $u$ vanish in the scaling limit, and thus for any $ W \in \R $ and any $ \varepsilon > 0 $:
\be\label{eq:revclaimii}
\lim_{M\to \infty} \mathbb{P}\left( \right| R_M(\pm W,L_M)  \left| > \varepsilon \right) = 0 \, . 
\ee
Combined with the interpolation bound of Theorem~\ref{thm:linear}  this implies that also $ \max_{|u|\le W} |R_M(u,L_M)|$ converges in distribution to  $0$.\\

\noindent {\it iii.a)}   ($\lambda < \sqrt 2$) It suffices to consider the case $\mathcal E_M  = \widehat{\mathcal E}_{M}^{(-1)} $, for which $\tau =0$, since \eqref{eq:iiia}  differs from it by just a microscopic shift by $\tau$.   In this case,  
\be 
\mathbb E \left[ R_M(0) \right]  \ =\   0 \, . 
\ee  
As in (ii), the bounds of Lemma~\ref{lem:tail_var} imply that for $\delta_M$ of \eqref{eq:Riia}  the fluctuations of $R_M(\pm\delta_M) $ are negligible (in probability) with respect to the mean, which is determined through $\frac{d}{d \,u} \mathbb E[R_M(\pm\delta_M) ] $.  Consequently, for any $ b > 0 $:
\be  \label{ii_delta}  
\lim_{M\to \infty} \mathbb{P} \left(   
  \frac{1}{\pm \delta_M} \, R_M(\pm\delta_M)    
  \, > \, b  \,   \  \frac{M^{2(\mathcal E_M/\lambda)^2 -1} }{\ln M  }   \right)  \ = \ 1  \\[1.5ex]  
\ee 
By the uniformity principle which is expressed in  Theorem~\ref{thm:linear}, this bound extends to \eqref{eq:Riia}. \\  

 \noindent {\it iii.b)}  ($\lambda > \sqrt 2$)   In this case  \eqref{eq:iiib} and \eqref{eq:HTrafo} imply
\begin{eqnarray} 
\lim_{M\to \infty} \mathbb E\left[  R_M(0)  \right]  & = &
\lim_{M\to \infty} M \Delta_M \left( \widehat \varrho_M(\mathcal{E}_M) -1 \right)  \\ 
& = & \lim_{M\to \infty} M \Delta_M \left( \widehat \varrho_M(\mathcal{E}_M) - \widehat \varrho_M(  \widehat{\mathcal E}_{M}^{(-1)} ) \right)   = \alpha  \, . \notag 
\end{eqnarray} 
For this case the bounds of Lemma~\ref{lem:tail_var} guarantee that both the variance and the mean derivative of $  R_M(u)$ make no contribution to the limiting behavior  of $  R_M(\pm W)$.  The monotonicity of the function allows to turn  this into the uniform statement which was claimed in \eqref{eq:iiib}. 
\end{proof} 

\medskip 

 \section{The main result} \label{sec:main}
 
 We now turn to the implications for the spectrum of $H_M$ in~\eqref{eq:H}.  
 
   \subsection{Band of semi-delocalized states} 

 Existence of localized states for this operator  has been noted before, and proven in a number of works~\cite{molchanov,SUSY}.  In this respect the new result presented here is that the operator also has bands of semi-delocalized states.   These  appear   in the vicinity of the spectrum of the mixing term $\ket{\phi_0} \bra{\phi_0}$, which is the two point set $\{-1,0\}$.    However, at $\mathcal E = -1$, the formation of a band of semi-delocalized states, rather than a single extended state,  requires  $\lambda > \sqrt 2$.

\begin{theorem}[Semi-delocalization]  \label{thm:resdeloc_rev}
Let  $\mathcal{E}_M$ be a sequence  
with 
\begin{enumerate} [i.]
\item $\lim_{M\to \infty}  \mathcal{E}_M  = 0 $  satisfying  \eqref{eq:ii} (at arbitrary $\lambda>0$),  or
 \item  $ \lim_{M\to \infty}  \mathcal{E}_M  = -1 $  satisfying  \eqref{eq:iiib} and  $\lambda > \sqrt 2$.  
\end{enumerate} 
Then, the behavior within the scaling window is as follows.   
\begin{enumerate}    
 \item The joint process of rescaled eigenvalues and potential values  converges in distribution  to the  \v{S}eba  point process at the level $0$ in case {\it (i)} and  $\alpha $ of \eqref{eq:iiib} in case {\it (ii)}.  \\ 
\item The moments of the corresponding eigenfunctions  behave as those of the \v{S}eba process (Lemma~\ref{lem:Seba_ef}) 
 in the  sense that  for any  $|n|<\infty$:  
\begin{enumerate}[a.] 
\item  for any $\epsilon > 0$ there is $\delta = \delta(\epsilon, n) > 0$ such that for all $M<\infty$
\be \label{ratio_i}
\mathbb{P} \left( 1 + \delta \ <\  \frac{\|\psi_{M,n}\|_2}{\|\psi_{M,n}\|_\infty }\  < \  \frac{1}{\delta} \right) \ >\  1 - \epsilon~.  
\ee
\item   for any $b<\infty$: 
\be \label{ratio_ii_0}
\lim_{M\to \infty} \mathbb{P} \left( \frac{\|\psi_{M,n}\|_1}{\|\psi_{M,n}\|_\infty}\  > b\   \right) \ =  \  1
\ee   
\end{enumerate} 
\end{enumerate}
  \end{theorem}

\medskip

To facilitate comparisons of the eigenfunctions in the different regimes discussed in this section, we adapt the following conventions.  
\begin{enumerate}  
\item  The  spectrum is obtained  from the characteristic equation \eqref{eq:sec}, for which  the cutoff parameter is set here to  $L_M = \ln M$.   
\item The  eigenfunctions  $\psi_{M,n}$ will be assigned the normalization \eqref{norm}, so that for each eigenvalue $u_{M,n}(\omega_M)$:  
\be  \label{Phead}
\| \psi_{M,n}\|_\infty \ = \  \frac{1}{\dist ( u_{M,n}(\omega_M), \omega_M)} \,.  
\ee     
\item The   functions' $\ell^2$-norms will be split  into the sum of the `head', `body', and `tail' terms:  
\be  \label{3split} 
 \| \psi_{M,n}\|^2_2  \ = \  \| \psi_{M,n}^{(H)}\|^2_2  + \| \psi_{M,n}^{(B)}\|^2_2   + \| \psi_{M,n}^{(T)}\|^2_2  \, , 
\ee     
 with:
 \begin{align}    \label{3terms} 
(i) \quad  \| \psi_{M,n}^{(H)}\|^2_2  &  := \  \frac{1}{|u_{M,n}-\omega_{M,n^*}|^2} \  = \ \| \psi_{M,n}\|^2_\infty  
 \notag \\   
(ii) \quad   \| \psi_{M,n}^{(B)}\|^2_2 &  :=  \  \sum_{ k \neq n^*}   \frac{\indfct[|\omega_{M,k}|< \ln M]}{|u_{M,n}-\omega_{M,k}|^2}   
   \\ 
(iii) \quad  \| \psi_{M,n}^{(T)}\|^2_2 &   \  :=   \  \sum_{ k}   \frac{\indfct[|\omega_{M,k}|\geq \ln M]}{|u_{M,n}-\omega_{M,k}|^2}     = 
   \frac{d}{du}  R_{M,\omega}(u_{M,n}) \, .  \hfill \notag
   \end{align} 
where  $n^*$ is the value of the index $k$ for which   $\omega_{M,k}$  is closest to $u_{M,n}(\omega_M)$.   \\ 
\end{enumerate} 
 
 In estimating the eigenfunctions we shall use the information on the scaling limits of $T_M(u)$ provided for the different regimes by Lemma~\ref{lem:genlinlimit}, and also the following elementary estimate.   
 
\begin{lemma} 
Let  $\delta^+_W(\omega_M) $ and $\delta^-_W(\omega_M) $   denote the 
the smallest, and correspondingly largest, gap between consecutive points of $\omega_M\cap (-W,W)$. 
Then for any $W>1$ and sequence $(\varepsilon_M)$ with $\lim_{M \to \infty} \varepsilon_M =0$: 
\be  \label{delta}
\lim_{M\to \infty} \mathbb{P} \left(  \delta^-_W  \geq \frac{\varepsilon_M}{W}\, ,  \quad  \delta^+_W  \leq   \frac{ \ln W}{\varepsilon_M}\,   \  \right) \ = \ 1 \,. 
\ee 
\end{lemma} 
\begin{proof} 
For fixed $W$ the bounds are implied by the  Poisson process calculation: 
\begin{align}    \label{delta_poisson}
& \lim_{M\to \infty}   
 \mathbb{P} \left(  \delta^+_W > \ (1+t) \, \ln W  \  \right) 
 \leq   \int_{-W}^W  \, dv \, e^{- (1+t) \, \ln W} \ = \  \frac{2 W }{ W ^{1+t} }  \, ,
  \notag \\     
& \lim_{M\to \infty}     \mathbb{P} \left(  \delta^-_W < \ s\,   \  \right) 
  \leq   \int_{-W}^W  \, dv \, s \,  \=\     2 W\, s ~ ,    
\end{align} 
for any $ s, t > 0 $
\end{proof}

\begin{proof}[Proof of Theorem \ref{thm:resdeloc_rev}] 
\noindent { \it 1.}   By Lemma~\ref{lem:genlinlimit}, parts {\it (ii)} and {\it (iii.b)}, for the two cases considered here 
$R_{M,\omega}(u)$ converges to constant functions, at the indicated values of the parameter $\alpha$.  
By Theorem~\ref{thm:GenCrit} this results in the convergence of  the joint distribution of $(u_{M,n},\omega_{M,n})$ 
to the corresponding \v{S}eba processes.    \\    

\noindent { \it 2.}  Turning to the eigenfunction, we shall discuss separately each of the three terms in 
\eqref{3split}. 

\noindent {\it (i)+(ii) \/}  By the distributional convergence of the of $(u_{M,n},\omega_{M,n})$  to the \v{S}eba process $ (u_n,\omega_n) $, the head 
contribution converges in distribution to the one of the \v{S}eba eigenfunctions $ 
\| \psi_{M,n} \|_\infty = \| \psi_{M,n}^{(H)}\|_2 \to \| \Psi_{n} \|_\infty $. Lemma~\ref{lem:Seba_loc}  and \eqref{delta} hence imply that for any $ W < \infty $:
\begin{align}\label{head1}
\lim_{M\to \infty} \mathbb{P}&  \left(\mbox{for all $n$ with $  u_{M,n} \in (-W,W)$:}    \right. \notag \\
& \left. \qquad \frac{1}{\gamma_M \ln W}  \leq      \| \psi_{M,n}\|_\infty \leq \ \gamma_M W  \right) 
= \ 1 
\end{align}
provided $\gamma_M \to \infty $.  

To estimate the `body' term, we note that for each $k$ the sum is typically comparable to $\sum_{k\neq 0} 1/k^2$, and hence finite.   For a bound on the maximum for energies in the range  $u_{M,n} \in (-W,W)$: \be   \label{gap}
   \frac{1}{(\delta^+_{W})^2} \ \leq \ \| \psi_{M,n}^{(B)}\|^2_2   \  \leq \     \frac{  4}{(\delta^-_{W+1})^2}\sum_{k \in \Z \backslash \{0\}}  \frac 1{k^2}    \ 
    +  \sum_{k } \frac{\indfct [  |W| \leq  |\omega_{M,k}| < \ln M]}{|\omega_{M,k}|^2}  ~.  
\ee 
Applying  \eqref{delta} to the first sum, and a simple Chebyshev bound to the second, we conclude that for each $W <\infty$:  
\begin{align}\label{body}
\lim_{M\to \infty} \mathbb{P}&  \left(\mbox{For all $n$ with $  u_{M,n} \in (-W,W)$:}    \right. \notag \\
& \left. \qquad \; \frac{1}{\gamma_M \ln W}  \leq    \| \psi^{(B)}_{M,n}\|_2
  \, \leq  \,   \gamma_M \,  W      \right)\ = \ 1  
\end{align}  
provided $\gamma_M \to \infty $.

\noindent {\it (iii) \/}  The convergence of $R_{M,\omega}(u)$ to constant functions  is accompanied (e.g., by the asymptotic linearity which is expressed in Theorem~\ref{thm:linear}) with the convergence of  
$  \frac{d}{du}  R_{M,\omega}(u_{M,n}) $ to $0$, uniformly over compacta.  In view of \eqref{3terms}, this in turn implies that the tail contribution to the functions vanishes, in the sense that for any $W< \infty$ and $\varepsilon >0$:    
\be \label{ratio_ii}
\lim_{M\to \infty} \mathbb{P} \left( \max_{n: \, |u_{M,n}|<W}  \|\psi^{(T)}_{M,n}\|_2 \  >  \ \varepsilon  \   \right) \ =  \  0
\ee   
 
The claim {\it 2.a}  is implied by the estimates \eqref{head1},  \eqref{body}, and  \eqref{ratio_ii}. 
 
\noindent  { \it 2.b} For the limiting \v{S}eba  process the $\ell^1$ norms, which are given by sums comparable to $\sum_{n\neq 0} 1/|n|$, diverge (as in Lemma~\ref{lem:Seba_loc}).   By Fatou's lemma, this implies divergence in distribution of $\|\psi_{M,n}\|_1/\|\psi_{M,n}\|_\infty \to \infty$ for $M\to \infty$.    (The more explicit bound indicated in the above remark for case {\it ii\/ } is based on $M \Delta_M$.)  

\end{proof}

\subsection{The localization regime}

At  energies which do not lie in the spectrum of the  mixing term in the Hamiltonian $H_M$, i.e. for $ \mathcal E  \notin \{ - 1 , 0 \}$, one finds  localization  irrespective of the value of $ \lambda$.  \\  

As was mentioned above, localization for the operator $H_M$ was discussed previously, starting with \cite{molchanov}.  
In terms of the method presented here, localization is associated to the singular ($\pm\infty$) limiting behavior of the tail contribution $R_{M,\omega}$ to the characteristic equation, in the terminology of Definition~\ref{def:3types}.   Following is a more quantitative analysis.

\begin{theorem}[Localization]\label{thm:localization}  For any sequence with a limiting value   
\be 
\label{eq:not10} 
\mathcal E \ = \ \lim_{M\to \infty} \mathcal E_M   \   \in (-\lambda,\lambda)  \backslash \{-1, 0\} \, , 
\ee 
 the eigenvalues within the scaling window behave as follows.  
  \begin{enumerate} 
 \item Under the scaling \eqref{eq:Escale} \& \eqref{eq:potpoint} the eigenvalues  coincide  asymptotically, in probability, with the  point process $(\omega_{M,n})$ of the rescaled potential values (the two being compared within scaling windows of fixed, but arbitrarily large size $[-W,W]$).  \\ 
\item The eigenfunctions corresponding to energies with $|u_{M,n}| < W$ (with $ W < \infty $ fixed but arbitrary) are all $ \ell^2$-localized in the sense that with asymptotically full probability all satisfy, for any $\gamma>0$:  
\be\label{eq:twobumps}
1\  \  \leq \  \frac{\| \psi_{M,n} \|_2}{\| \psi_{M,n} \|_\infty}  \ \leq  \  
  \ 1   +     
       \mathcal{O}\left(\frac{1}{M^{(\mathcal E/\lambda)^2 (1-\gamma)}} \right)   
       +        \mathcal{O}\left( \frac{1}{M^{1/2 \, - \gamma}}\right)    \\[2ex] 
\ee 
\end{enumerate}  
\end{theorem}

\noindent {\bf Remarks:} 
 The somewhat vague statement in {\it (2)\/}  is made more explicit by the bounds provided in the proof.)
  
 As an immediate corollary of {\it (1)\/}   we conclude that in the regime covered by Theorem~\ref{thm:localization} the rescaled eigenvalue process~converges  in distribution to a Poisson process. 

\begin{proof}  
\noindent { \it 1.}       
The asymptotic coalescence of the rescaled spectrum with $\omega$ is a direct consequence of the divergence of $R_M$, in the sense of Lemma~\ref{lem:genlinlimit} part (i), and the statement proven in Lemma~\ref{lem:Seba_loc} (for which the value of the parameter $\alpha$ diverges due to the singular nature of the limit of $R_M$).\\  

\noindent { \it 2.}   
In the regime discussed here, due to the  asymptotic coalescence of the spectrum with $\omega$,  the `head' term $\| \psi_{M,n}^{(H)}\|_\infty$ diverges, typically at the rate:  
$
    \| \psi_{M,n}^{(H)}\|^2_2   \  = \  \| \psi_{M,n}^{(H)}\|^2_\infty  \ \geq \  M^{2 (\mathcal E_M/\lambda)^2 (1-\gamma) }   $,  
 where $\gamma > 0$ can be taken arbitrarily small. 
 More explicitly,   
the bounds of Lemma~\ref{lem:genlinlimit} combined those of with Lemma~\ref{lem:Seba_loc} allow to conclude that 
 for each $W <\infty$, $\gamma>0$ and $b<\infty$
\be  \label{psi_infty}
\lim_{M\to \infty} \mathbb{P} \left( \min\left\{   \| \psi_{M,n}\|_\infty  \, :  u_{M,n} \in [-W,W] \right \}    
  \, \geq  \, b    \   M^{ (\mathcal E_M/\lambda)^2(1-\gamma)}  \right)  \ = \ 1 \, . \\[2ex]   
\ee

For the `body'  term, $\| \psi_{M,n}^{(B)}\|_2$, which depends mainly on the gaps in $\omega$, the bounds \eqref{gap} and  \eqref{body} apply with no change  in  the different spectral regimes considered in this work.  \\

The  `tail' term  does not vary by more than a factor $1+\mathcal{O}(1/\ln M)$ among  all the eigenfunctions within the window $[-W,W]$, cf.~\eqref{eq:vergle}.
One may note that
\be 
 \| \psi^{(T)}_{M,n}\|^2_2  \ = \  \frac{d}{du}  T_{M,\omega}(u_{M,n}, \ln M) \, . 
\ee 
Applying the bounds of Lemma~\ref{lem:tail_var} and \eqref{eq:vergle} one readily gets: 
\be  \label{tail}
\lim_{M\to \infty} \mathbb{P} \left( \max_{n:  u_{M,n} \in [-W,W] }    \| \psi^{(T)}_{M,n}\|^2_2  \,  
  \, \leq  \, b    \   \left[1 + M^{2(\mathcal E_M/\lambda)^2 -1}   \ln M \right]   \right)\ = \ 1  \\[2ex]   
\ee  
 for each $W <\infty$,  and $b<\infty$.    The power law in this bound  can be made intuitive by noting that 
the   contribution  to the sum from sites with regular values of $V(x)$ is itself about 
$M\Delta_M^2 \approx M^{2(\mathcal E_M /\lambda)^2 -1}/ \ln M$.  \\  

The estimates, \eqref{psi_infty}, \eqref{body}, and \eqref{tail} directly imply the claim \eqref{eq:twobumps}. 
In essence these bounds show that in the $\ell^2$-sense 
 the  `head' contribution dominates throughout the regime discussed here.  One may also note that among the other two terms the dominant one is the `tail'  for  $\lambda / \mathcal E < \sqrt 2 $, and the  `body' for $\lambda/ \mathcal E > \sqrt 2$. 
\end{proof} 

\medskip

\subsection{Localization alongside an isolated extended state near $- 1$}\label{sub:near-1}

The exclusion of  $ \lambda \le  \sqrt{2} $ in the second part of Theorem~\ref{thm:resdeloc_rev}  is relevant.  Intuitively, at $\lambda =0$  the operator $H_M$ has  a single extended state ($ \phi_0$), which at $\lambda \approx 1$ starts to be  passed, though `avoided crossings' by a series of localized states.   The next result shows that the picture of a single extended state embedded among localized states persists for $\lambda < \sqrt 2$ (except for instances of hybridization during the avoided crossing which are too brief to be seen is the single $\lambda$ snapshots that are discussed here).  In Theorem~\ref{thm:resdeloc_rev} we saw that this picture changes at  $\lambda =  \sqrt 2$, beyond which the operator acquires  a band of extended states, with energies in the vicinity of $\mathcal E = -1$.  \\

\begin{theorem}[Non-resonant delocalized state] \label{thm:localizationplusdeloc} 
For $\lambda < \sqrt 2$,  let $\mathcal E_M$ be a sequence of energies satisfying $ \lim_{M\to \infty}  \mathcal{E}_M  = -1 $  and the condition \eqref{eq:iiia}.  Then, within the scaling windows centered at  $\mathcal E_M$:
\begin{enumerate}
\item 
There exists one eigenvalue, which occurs at (microscopic) energy   
\be\label{eq:onedeloc}
u \ = \ \tau +  o(1) 
\ee
for which the corresponding eigenfunction $ \psi_E $ is $ \ell^2$-delocalized, with 
\be  \label{1_infty}
\lim_{M\to \infty} \mathbb{P} \left( \frac{\| \psi_E \|_2}{  \| \psi_E \|_\infty} \geq  
M^{\frac{1}{\lambda^2} - \frac{1}{2} - \gamma}  \right)  \ = \ 1  \, ,   
\ee  
for any $\gamma>0$. \\ 
\item
 All other eigenfunctions  in the scaling window are $ \ell^2$-localized in the sense that for each $W<\infty$ and $\gamma>0$: 
 \be  \label{tail_0}
\lim_{M\to \infty} \mathbb{P} \left( \max_{n:  |u_{M,n}|<W }    
\frac{\|\psi_{M,n}\|_2}{\|\psi_{M,n}\|_\infty} \leq 1 +  \frac{1}{ M^{\frac{1}{\lambda^2} - \frac{1}{2} - \gamma }}   \right)\ = \ 1  
\ee 
\end{enumerate}
\end{theorem}

\begin{proof}
The behavior of the function $R_{M,\omega}(u)$ in this case, is described by \eqref{eq:Riia} of Lemma~\ref{lem:genlinlimit}: 
\begin{enumerate} 
\item the function is monotone and of high derivative, its  typical order being 
\be \frac{d}{du} R_{M,\omega}(u)  =  \mathcal{O}\left( \frac{M^{2(\mathcal E_M/\lambda)^2 -1} }{\ln M  } \right) ~,
\ee 
\item  $R_{M,\omega}(u) =0 $  at $u = \tau + \mathcal O(\delta_M)$, with  $\delta_M= (\ln M)^2 / M^{(\mathcal E_M/\lambda)^2 -1/2}$, 
   \item the value of  $S_{M,\omega}(\tau,\ln M)$ does not depend on $\tau$, has the Cauchy distribution, and is typically of order $1$. 
 \end{enumerate} 
It follows  that the characteristic equation, $S_{M,\omega}(\tau,\ln M) = - R_{M,\omega}(u)$ has one solution at 
\be 
u \ =\  \tau +   \mathcal O(\delta_M)
\ee 
and others at close proximity to the poles of $S_{M,\omega}(\tau, \ln M )$, as described in Theorem~\ref{thm:GenCrit}.  
The characteristics of the eigenfunctions can be read off this description by following the arguments which were used in the proofs of Theorems~\ref{thm:localization} and \ref{thm:resdeloc_rev}. 
\end{proof}

\section{Discussion   }   \label{sec:disc} 

\subsection{Relation with a two-state hybridization criterion}

While our main results concern the effects of resonances involving many localized approximate eigenfunctions, let us   note their relation with a  simple criterion for two level eigenfunction hybridization.     

In the simplest two level system the Hamiltonian $H$ and the corresponding resolvent operator $G(\zeta) := (H-\zeta \id)^{-1}$ are of the form of the $2\times 2$ matrices
 \be 
H \ = \ 
\begin{pmatrix} \mathcal E_1 & \tau  \\ \tau^*  &   \mathcal E_2 
    \end{pmatrix} \,  , 
    \qquad  
    G(\zeta)  \ = \ 
\left[ \, \begin{pmatrix} \mathcal E_1-\zeta  & \tau  \\ \tau^*  &   \mathcal E_2 - \zeta 
    \end{pmatrix} \, \right]^{-1}. 
\ee 
The spectrum and eigenfunctions can be found by studying the poles and residues of the resolvent matrix $G$. 
Of particular interest to us is the case where $\tau $ is small.   In this situation, the system has two approximate 
eigenstates, or quasi-modes, corresponding to the (column) vectors $(1,0)$ and $(0,1)$, with $\tau$ serving as the mixing term, or the tunneling amplitude.   The relevant quantity is the ratio of the quasi-modes' energy gap   $\Delta \mathcal E = (\mathcal E_2-\mathcal E_1)$ to the tunneling amplitude.  A simple calculation shows that:
\begin{itemize} 
\item[i.]   If $ |\Delta \mathcal E |\gg |\tau|$  then the eigenfunctions of $H$ are localized, i.e., $ \Psi_1\  \approx \  \left(1,\,  0 \right) $, $ \Psi_2 \ \approx \ 
\left(0, \, 1 \right) $,  to the leading order in $\tau/|\Delta \mathcal E |$;
\item[ii.]    If  
$ |\Delta \mathcal E |\ll |\tau|$ 
 then the eigenfunctions are  equidistributed between the two sites, and close to:  
$
\Psi_1\  \approx \  \frac{1}{\sqrt 2} \left(1,\,  1 \right) $, 
$ \Psi_2 \ \approx \ 
  \frac{1}{\sqrt 2} \left(1,\,  -1 \right) $. 
 \end{itemize} 

Turning  to a system with a large configuration space and random potential, the following is a useful and established term.       
\begin{definition}[Quasi-modes]  
A \emph{quasi-mode} for the self-adjoint operator $H$ with discrepancy $d$ is a pair $(\Psi, \mathcal E)$ such that 
\be 
\| (H - \mathcal E) \Psi \| \ \le \  d \, \| \Psi \|  \, . 
\ee 
\end{definition} 

Tunneling amplitude is a regularly used term 
however its meaning is often left somewhat open, allowing for creative interpretation.   In the context of operators with on-site disorder, we find the  following formulation to be of relevance.   

\begin{definition}[Tunneling amplitude]  \label{def:tunneling_amp}
For a collection of orthogonal quasi-modes $(\Psi_j, \mathcal E_j)$, with $\Psi_i \perp \Psi_j$ at $i\neq j$ and $P_j$ the corresponding projections, we define the  pairwise \emph{tunneling amplitude} as  $|\Sigma_{i,j}(\mathcal E)|$ the modulus  of the off diagonal term in the following representation of the operator's resolvent, $(H - \mathcal E)^{-1}$, restricted to the range of $P_i+P_j$: 
\be \label{eq:quasimode}
\begin{pmatrix} G_{i,i}(\mathcal E )  & G_{i,j}(\mathcal E )  \\ G_{j,i}(\mathcal E )  &  G_{j,j}(\mathcal E )     \end{pmatrix}  
\ = \ 
\left[ \, \begin{pmatrix}  \mathcal E_i  + \Sigma_{i,i}( \mathcal E)  &  \Sigma_{i,j}(\mathcal E) \\ \Sigma_{j,i}( \mathcal E)     &      \mathcal E_j  + \Sigma_{j,j}(\mathcal E) 
    \end{pmatrix} \, \right]^{-1}
    \ee 
\end{definition} 

\medskip 

To place that in context, let us recall  the Schur complement formula, which states that if for a specified pair $(i,j)$, the operator $H$ is decomposed as
$
H  = \mathcal E_i \, P_i  +  \mathcal E_j \, P_j  +   \widehat H  $,
then~\eqref{eq:quasimode} holds with $\Sigma(\mathcal E) = \begin{pmatrix}     \Sigma_{i,i}  &  \Sigma_{i,j} \\ \Sigma_{j,i}     &   \Sigma_{j,j} 
    \end{pmatrix} $  the $2\times 2$ inverse of the restriction of $(\widehat H - \mathcal  E)^{-1}  $ 
to  the range of $ P_i+P_j$.  \\

Consider now the situation in which an $M\times M $ matrix $ H_M $ has a large collection of quasi-modes, whose energies fluctuate with a considerable degree of independence at density $\mu_M(\mathcal E)$. Thus in the 
vicinity of energy $ \mathcal{E}$ they have gaps of the order 
$
\Delta_M (\mathcal E) = \ [M\, \mu_M( \mathcal E)]^{-1}  $.  
Assume also that the pairwise tunneling amplitudes are ``mostly''  of a common order of magnitude $\tau(M, \mathcal E)$.    
   The  previous rank-two discussion suggests, as a `rule of thumb' that at energies at which  
\be   \label{eq:DeltaU_tau}  
\frac{\Delta_M ( \mathcal E) }{ \tau(M, \mathcal E) }\  \lesssim \ 1\,    
\ee 
the localized quasi-modes would be unstable with respect to resonant delocalization, and the proper eigenstates will take the form of hybridized wave functions.    This paper grew out of an attempt to develop further insight on the relevance of the  criterion~\eqref{eq:DeltaU_tau}.  
We find that our results support its relevance in the present context.    \\ 

 More explicitly:  from  the expression~\eqref{eq:res} for the resolvent, we find that for the operator $H_{M}$ the tunneling amplitude between a pair of  the $\delta$ function quasi-modes is given by
 \begin{eqnarray}  \label{eq:tauij}
\tau_{i,j}(\mathcal{E}) & = &   \frac{1}{M} \left| 1 + \langle \varphi_0 \, , \big(\widehat H_M - \mathcal E  \big)^{-1} \varphi_0 \rangle \right|   
\notag \\ 
& =&  
 \frac{1}{M}  \left| 1-   \frac{1}{M} \sum_{n\neq i, j} \frac{1}{\kappa_M V_n - \mathcal E }  \right|^{-1} 
 \end{eqnarray} 
where $ \widehat H_M $ is a modified version of $ H_M $ with $ V_i = V_j = 0 $.  It  may be note  that the `direct' tunneling amplitude $ 1/M $ is boosted by the term $\langle \varphi_0 \, , \big(\widehat H_M - \mathcal E  \big)^{-1} \varphi_0 \rangle$.   Intuitively, that is so since the tunneling is through the state $\varphi_0$.    \\

By the estimates of  Lemmas~\ref{lem:tail_var} and \ref{lem:meanT},  typical value of the tunneling amplitude among states of energy in the vicinity of 
$\mathcal E$ is: 
\begin{eqnarray}  \label{eq:tauij_est}
\tau_{i,j}(\mathcal{E}) 
& \approx& 
 \frac{1}{M}    \left[  1-  \widehat{\varrho}_M( \mathcal E)    +   
 \Theta \left(\frac{1}{M \, \Delta_M(\mathcal E)} \right ) \right] ^{-1}  \, .
\end{eqnarray} 
Thus, the heuristic condition for resonant delocalization \eqref{eq:DeltaU_tau} corresponds to:   
\be  \label{eq:resdeloccond}
M\,  \Delta_M( \mathcal E)  \left| 1 - \widehat{\varrho}_M( \mathcal E) 
+\Theta\left( \frac{\sqrt{\ln M}} {\lambda \, M^{(\mathcal E/\lambda)^2 }}\right) \right|     \  \lesssim    \  1  
  \, .    
\ee 
which played a role in our discussion in Section~\ref{sec:regions} (see  \eqref{eq:Tmean}).  \\

To answer the question posed above we note that  the two-level resonance condition is pointing at the right direction.  At the same time,  rank-two analysis alone does not yet address  a number of relevant points, such as: 
 \begin{enumerate}[i.] 
 \item a possible stratification of quasi-mode pairs, e.g. by the distance  (which in the tree graph example studied in \cite{AW13} affects the tunneling amplitude) or by the efficacy of mixing channels, 
\item the effects of possible  interactions among distinct  quasi-modes, as well as other states 
\item the question of formation of a \emph{band} of extended (or semi-extended) eigenstates.  \\ 
 \end{enumerate} 
 The last   has been tackled here through  the  more detailed analysis of the structure of the resolvent, and the  Green function's local scaling limit.   \\

We also found that the hybridization  among levels which resonate through a single channel, as in the  present example, yields   delocalization in only a partial sense: it is delocalization in the spacial distribution of the wave function and in the $\ell^1$-sense, 
  but not in the $\ell^2$-sense,  meaning that most of the state's $\ell^2$-mass is carried  by few   localized 
sites.  This point is discussed next.  \\

\subsection {Different notions of delocalization} \label{sec:localization}

Localization and delocalization of the eigenfunctions on a finite or infinite graph can be formulated in terms which may either 
depend on the graph's metric or be independent of it.   It is a relevant observation that the two  terminologies need not coincide.  
The point is exemplified by functions which are the sum of  few localized wave packets, which are located at large distance from each other.   Judging by their spatial spread,  such functions would be deemed {\it delocalized}, whereas judging by the number of points on which the bulk of the corresponding probability distribution (or $\ell^2$-norm)   is supported, the function may be viewed as {\it localized}.  
This distinction is of relevance in the case discussed below, and  more generally wherever  eigenfunctions are 
formed through resonances among local quasi-modes.   \\

The two different  forms of localization for functions on a graph $\mathcal{G}$ can be quantified through the quantities: 
\begin{align}
p_{diam}(d) &=  \inf \Big\{ \sum_{x\in \G \backslash A}  \frac{|\psi (x)|^2}{\|\psi \|_2^2}  \, \big|  \, 
A\subset \G,\, \rm{diam}(A) \le d \Big\}  \notag \\ 
p_{vol}(d) &= \inf \Big\{ \sum_{x\in \G \backslash A}  \frac{|\psi (x)|^2}{\|\psi \|_2^2}  \, \big|   \, 
A\subset \G,\, \rm{card}(A) \le d \Big\} 
\end{align}  
with $\rm{diam}(A)$ the diameter of the set $A$,   $\rm{card}(A)$ the set's cardinality.\\

\begin{enumerate} 
\item[]{\bf Spatial localization}   can be expressed through suitable bounds on the  distances at which $ p_{diam}(d)$ reaches small values.    For example, `spatial exponential localization' with localization length $\xi$ may be expressed by a bound of the form 
\be
p_{diam}(d) \ \le \  C \, e^{-d/\xi}  \, , 
\ee 
 at some $C<\infty$.  \\ 
 
 \item[]{\bf $\ell^2$-localization}  is similarly expressed through bounds on the inverse function of  $p_{vol}(d)$.    Exponential  $\ell^2$-localization would be expressed by a bound of the form
\be
p_{vol}(d) \ \le \  C \, e^{-d/\alpha}  \, , 
\ee 
 with some $C, \alpha <\infty$.  
 The inverse of that  function show how many sites does it take to capture all but  fraction $p$ of the function's $\ell^2$-mass.\\ 
    \end{enumerate} 
 In case the discussion concerns not a fixed graph but a sequence of graphs, of diverging diameters, the two notions of  localization may be tested by whether the inverse functions of $p_{diam}(d)$, and correspondingly $p_{vol}(d)$ are uniformly bounded, or at least grow at slow rate.  \\ 

In the converse direction, the term delocalization  can also be given  different meanings:   
\begin{enumerate} 
\item[]{\bf Spatial delocalization} on scale $\xi_M$, can be expressed by the condition that eigenfunctions with energies in the specified range  typically satisfy:  
\be  \label{eq:spatial_deloc}
\sum_{x,y\in \mathcal{G}_M} |\psi(x)|^2 \, |\psi(y)|^2\,  \indfct\left[ {\rm dist}(x,y) \geq {\xi_M} \right]   \  \ge \   p_0 \,   \|\psi\|_2 ^2 \, . 
\ee
for some $M$-independent $p_0 > 0$. (Here $\mathcal{G}_M$ are graphs $\G_M$ of growing diameter).\\ 

\item[] {\bf  $\mathbf{\ell^2}$-delocalization}  
for a sequence of functions on graphs $\G_M$ of growing diameter,   
{$\ell^2$-delocalization} is expressed in the statement that   
\be \label{eq:deloc}
\lim_{M\to \infty} p_{vol}(d) \ = \ 1 
\ee 
for  all $d<\infty$. \\ 
\end{enumerate}%

A more standard formulation of delocalization is  through the  function's {\bf inverse participation ratio} (with $q=2$ or more generally $q>1$)
\be  
P_q(\psi) := \frac{\sum_x |\psi(x)|^{2q}}{\left[ \sum_x |\psi(x)|^2 \right]^q} \,   .
\ee 
These are linked to the norm ratio $
r(\psi)  :=  \|\psi \|_\infty / \|\psi \|_2   \,  
$
(with $\|\psi\|_\infty  := \max_{x\in \mathcal{G}_M} |\psi(x) |$) through the bounds: 
\be  \label{eq:equiv}
 r(\psi) ^{2q} \ 
\le \   P_q (\psi)
\  \le \    r(\psi) ^{2(q-1)}  \, . 
\ee 
(The upper bound is implied by a convexity argument, and the lower bound is due to   the contribution of the site at which $\psi$ is maximized.)

Thus  {\bf $\mathbf{\ell^2}$-delocalization}~\eqref{eq:deloc} is equivalently expressed in the vanishing, in the suitable sense, of the inverse participation ratio $P_2$, or  of the norm ratio $r(\psi) $.   In dynamical terms this can be viewed as the opposite of 
{\bf positive recurrence} (terminology which is suggested by an analogy with a classical Markov chain term).

The example considered here   is  degenerate, since  on the complete graph the distance between  neighboring sites equaling the graph's diameter (i.e., except for the situation of   total localization - when the $\ell^2$-mass is asymptotically concentrated at a single site, the function's support is of diameter comparable with that of the entire graph). 
Thus, in Theorem~\ref{thm:resdeloc_rev}, we find that resonant delocalization occurs in the sense of spatial delocalization without meeting the standard $\ell^2$-delocalization criterion.   However, this partial delocalization does coincide with the following  weaker  measure of the spread of the wave function \\ 
 \begin{enumerate} 
\item[]{\bf $\ell^{1}$-delocalization:} as the size of the system is taken to infinity, for all $ q \in (0,1/2] $, the ratio $ P_q(\psi) $ diverges in the distributional  sense (for eigenfunctions with energies in a specified range). \\ 
 \end{enumerate} 
  
\medskip

\appendix

\section{The spectral range of $H_{M,\omega}$ and the ground state transition}   \label{app:ground_state}

This appendix is devoted to a proof of Theorem~\ref{thm:spectrumandgs}. We abbreviate $E_0 = \min \sigma(H_M)$. Since $T$ is of rank one, the eigenvalues of $H_M$ and $\kappa_M V$ interlace. Since $T \leq 0$ and by the extreme-value statistics of Gaussian random variables, 
\be\label{eq:extremevaluestat}
\lim_{M \to \infty } \mathbb{P}\left( \max_x V(x) > \frac{\lambda}{\kappa_M } - \frac{\ln( 4\pi \ln M)}{2 \sqrt{2 \ln M}} + \frac{u}{ \sqrt{2 \ln M}}  \right)  \to 1 - \exp\left(-e^{-u}\right) \, , 
\ee
for any $ u \in \R $ (cf.~\cite{Bovier}),  this implies that for any $ \varepsilon > 0 $:
\be\label{eq:Hausdconv2}
\lim_{M\to \infty} \mathbb{P}\left( d_H( \sigma(H_M) \setminus \{E_0\}  \, , \,  [-\lambda, \lambda] ) > \varepsilon \right) \ = \ 0 \, .
\ee
Thus all the spectrum, except for the ground state, converges as claimed. 
We now distinguish two cases.

\vspace{2mm}\noindent {\bf The case $0 < \lambda < 1$.} 
According to Proposition~\ref{lem1}, $E_0$ is the smallest
solution of the equation $  F_M(E) = 1 $.  For the remainder of the proof we may assume that 
\be\label{eq:maxass} 
  \max_x |\kappa_M V(x)| \leq \lambda 
 \ee 
since 
this holds with asymptotically full probability by~\eqref{eq:extremevaluestat}. Under this assumption,
for any $E < -\lambda $:
\be\label{eq:FM}
F_M(E) = \frac{1}{M} \sum_{x=1}^M \frac{\indfct[\kappa_M V(x) \geq - \lambda ]}{\kappa_M V(x) - E} = T_M\Big(E, - \frac{ \lambda + E}{\kappa_M}\Big)~.
\ee
This implies that for any $ E < -\lambda  $, any $ \eta > 0 $ and  all sufficiently large $M $:
\begin{equation}\label{eq:nother}
\lim_{M \to \infty }  \mathbb{P}\left( \left| F_M(E) - \widehat\varrho_{M}(E) \right| \geq M^{-1/2+ \eta}   \right) = 0    \, . 
\end{equation}
The proof of~\eqref{eq:nother} either follows directly from Lemma~\ref{lem:meanT}  and \ref{lem:tail_var}. Alternatively, it is  derived using the representation
$ F_M(E) = - (\lambda+ E)^{-1} - \int_{-\lambda}^\infty \frac{1}{M} \sum_x \indfct[\kappa_M V(x) \geq t ] \frac{dt}{(t-E)^2}  $ and  the Dvoretzky-Kiefer-Wolfowitz inequality \cite{DKW}:
\be \mathbb{P} \left\{ \sup_t \left| \frac{1}{M} \sum_x \indfct[\kappa_M V(x) \geq t ] - \int_{t/\kappa_M}^\infty e^{-s^2/2} \frac{ds}{\sqrt{2\pi}} \right| \geq \frac{R}{\sqrt{M}} \right\} \leq C e^{-2 R^2} 
\ee 
for any $R > 0$. 

Since $ \widehat{\mathcal{E}}_M^{(-1)} $ is the unique solution of $ \widehat\varrho_{M}(E)=1 $ for $ E \leq -1 < -\lambda $, the bound~\eqref{eq:nother}  implies that for any $ \eta > 0 $ with asymptotically full probability 
$ | E_0 - \widehat{\mathcal{E}}_M^{(-1)} | \leq C M^{-1/2+ \eta}  $. From~\eqref{eq:asympE} we hence conclude  $E_0 =  -1 - \kappa_M^2 + O(\kappa_M^4) $.

To prove the strong delocalization of the eigenfunction $ \psi_0(x) = (\kappa_M V(x) - E_0)^{-1} $  corresponding to $E_0$ in this case, we  use~\eqref{eq:maxass} to estimate:
\be
\sqrt{M} \geq \frac{\| \psi_0 \|_2 }{ \| \psi_0 \|_\infty} \geq \ \sqrt{M} \frac{1- \lambda}{2 ( 1 +\lambda)} \, .
\ee
This concludes the proof in case $ 0 < \lambda < 1 $. \\

\vspace{2mm}\noindent
{\bf The case $\lambda > 1$}. In this case, the statement is contained in Theorem~\ref{thm:localization} applied  with $ \mathcal E = - \lambda $. Alternatively, we may directly bound
$ E_0 \leq \min \sigma(\kappa_M V)   $. For a lower bound we may use the variational characterization and the Cauchy-Schwarz inequality:
\be\label{eq:minsigma}
E_0   \geq \inf_{ \| \psi \|_2 = 1 } \left[\sum_x \kappa_M V(x) \psi(x)^2 - \frac{1}{M} \sum_y \frac{1}{\alpha(y)} \sum_x \psi(x)^2 \alpha(x) \right] 
\ee
for any $\alpha > 0$.
We pick
\be
\alpha(x) = \begin{cases}
R + M(\kappa_M V(x) - \min \sigma(\kappa_M V)~, & V(x) \leq - R \\
\frac{M}{1-\delta}~, &V(x) > - R~,
\end{cases}
\ee
where $\delta > 0$ is such that $(1-\delta)^{-1} < \lambda$, and $R>0$ is a large number independent of $M$ to be chosen shortly. One can check that
\be
\sum_{V(x) \leq - R} \frac{1}{\alpha(x)}  \leq \delta/2
\ee
for sufficiently large $R>0$.  Therefore
$ \sum_x \alpha(x)^{-1} \leq 1 - \delta/2 $ which by~\eqref{eq:minsigma}, yields
\be E_0 \geq \min \sigma(\kappa_M V) - \frac{2R}{M}~. \ee
By the same argument, the minimum in (\ref{eq:minsigma}) is attained on a function $\psi_0$ which satisfies~\eqref{eq:gs}.

\medskip


\section{A  useful estimate}  \label{app:est}  

In Lemma~\ref{lem:tail_var} use was made of the following bound.    \\ 

\begin{lemma}\label{lem:intrho}  Let $\varrho(v) = (2\pi)^{-1/2} \exp(-v^2/2)$.  
Then for some $ 0 < c, C < \infty $ and any $0 < \delta \leq 1$ and any $v \in \mathbb{R}$:
\be \label{eq:est} 
\frac{c}{(1+|v|)^2} \leq  \int_{|u-v|\geq\delta} \frac{\varrho(u)du}{|u-v|^2}
\leq C \left\{ \frac{\varrho(v)}{\delta} + \frac{1}{(1+|v|)^2}\right\}~.
\ee
\end{lemma}
\begin{proof} 
Let us assume that $|v|\geq 1$; the proof is similar and simpler for $|v| \leq 1$. For the upper bound we decompose the integral into two parts. The first part is estimated using the
elementary inequality $\varrho(v+x) \leq \varrho(v) e^{|vx|}$:
\be
 \int_{\delta \leq |u-v| \leq C_1|v|^{-1}} \frac{\varrho(u) }{|u-v|^2} du \ \leq \ e^{C_1} \frac{\rho(v)}{\delta} \, . 
\ee
The second integral is bounded by
\be
\int_{|u-v|\geq C_1 |v|^{-1}} \frac{e^{-u^2/4}}{\sqrt{2\pi}} \frac{e^{-u^2/4}}{|u-v|^2} du \ \leq \ C_2 \max_{|u-v| \geq C_1 / |v|} \frac{e^{-u^2/4}}{|u-v|^2}~ \, . 
\ee
For suitably chosen $C_1 > 0$, the maximum is attained inside the interval $[-1/2, 1/2]$, and is
thus bounded by $C_3 / |v|^2$.

For the lower bound we estimate the contribution from $ | v | \leq 1 $:
\be
 \int_{\substack{|u-v|\geq\delta \\ | u | \leq 1 }} \frac{\varrho(u)du}{|u-v|^2} \geq \min_{| u | \leq 1} \varrho(u) \,  \int_{\substack{|u-v|\geq\delta \\ | u | \leq 1 }} \frac{du}{|u-v|^2} \geq \frac{c}{(1+|v|)^2} \, . 
\ee
\end{proof}

Of possible interest is also  the following exponential improvement of the variance bound of Lemma~\ref{lem:tail_var}.  

\begin{lemma}   \label{lem:tail_exp} 
For  $\lambda>0$, $\mathcal E \in (-\lambda , \lambda)$,  $|u| < L $, and any $\tau >0$
\begin{multline}  \label{eq:Tvar}
\qquad\mathbb{P}\left(\big| T_M(u,L)  - \mathbb{E}(T_M(u,L))   \big | \geq \tau  \right) \\ 
\leq 2 \exp\left( - c \tau \min\left\{ \frac{ \tau}{\mathbb{E} \left( \big | \frac{d}{d \, u}  T_M(u,L) \big |  \right)} , L  \right\} \right) 
\end{multline}
with some numerical constant $ c  > 0 $.
\end{lemma}
\begin{proof}  
The probability bound is based on the exponential Chebyshev estimate:
\be
\mathbb{P}\left(\big| X \big | \geq \tau  \right) 
\leq e^{-t \tau} \left(\mathbb{E} \left[ e^{t  X }\right]  + \mathbb{E} \left[ e^{- t  X}\right] \right) 
\ee
applied to $ X =  T_M(u,L)  - \mathbb{E}(T_M(u,L))  $, and optimized over $ t \geq 0 $.

In  a variant of the argument which was used in the proof of Lemma~\ref{lem:tail_var}, the moment generating function can be estimated by noting that $\mathbb{E} \left[ e^{t  T_M(u,L)  }\right]  $ is an average of a product of $M$ functions of iid random variables.    One obtains: 
\begin{align}  \label{eq:102} 
\mathbb{E} \left[ e^{t  T_M(u,L)  }\right] & = \prod_n \left( 1 + \mathbb{E}\left[ e^{t \frac{  \indfct \left[ |\omega_{M,n}| > L  \right]}{ \omega_{M,n}-u}} -1\right] \right) \\
& \leq \exp\left\{t \,  \mathbb{E}[T_M(u,L)]  + \frac{t^2}{2}  e^{\frac{|t|}{L}} \mathbb{E} \left[ \sum_{n}  \frac{  \indfct \left[ |\omega_{M,n}| > L  \right]}{ |\omega_{M,n}-u|^2 }  \right]  \right\} \, . \notag
\end{align}
where the inequality is based on the elementary bounds: $ e^x \leq 1  + x + \frac{x^2}{2} e^{|x|} $, for $ x \in \R $, and $   1 + x \leq e^x$, for 
$x \geq 0$. The second expectation value in the exponential equals
\be \label{eq:202} 
\mathbb{E} \left[ \sum_{n}  \frac{  \indfct \left[ |\omega_{M,n}| > L  \right]}{ |\omega_{M,n}-u|^2 }  \right]   =  \mathbb{E} \left[ \big | \frac{d}{d \, u}  T_M(u,L) \big |  \right] \, . 
\ee
The choice $ t = \min\left\{ L , \tau /  ( e \mathbb{E} \left[ \big | \frac{d}{d \, u}  T_M(u,L) \big |  \right]) \right\} > 0 $ in the Chebyshev inequality yields the claim \eqref{eq:Tvar}.  \\  
 \end{proof} 
 
For  explicit probability bounds  \eqref{eq:Tvar} may be combined with \eqref{eq:Tprime} of Lemma~\ref{lem:tail_var}. 
Although it was not used in the present work, this exponential bound is included here since it allows to strengthen the implications on the fluctuations of $T_{M,\omega}(u,L)$ into estimates which apply   uniformly over macroscopically broad ranges of $\mathcal E$. 
  
\bigskip

\subsection*{Acknowledgment}
This work was supported in part by NSF grant PHY-1104596.   MA and SW  thank  CIRM (Luminy)  and MS and SW thank the Institute for Advanced Study (Princeton)  for the support and hospitality  accorded to the authors during work on this project.

 \end{document}